\documentclass[12pt]{article}
\usepackage{amsmath}
\usepackage{graphicx,psfrag,epsf}
\usepackage{enumerate}
\usepackage{natbib}
\usepackage{url} 

\usepackage{amsmath,amssymb,amsfonts,amsthm,bbm}
\usepackage{verbatim}
\usepackage{setspace}
\usepackage{graphicx}
\usepackage{subcaption}
\usepackage{enumitem}
\usepackage{color}
\usepackage{bbm}

\newif\ifStat
\newif\ifNote

\Stattrue
\Notetrue

\ifStat
\renewcommand{\hat}{\widehat}
\renewcommand{\tilde}{\widetilde}
\renewcommand{\bar}{\overline}
\else

\fi

\makeatletter
\newcommand*{\rom}[1]{\expandafter\@slowromancap\romannumeral #1@}
\makeatother

\newcommand{\diag}{\mathrm{diag}}
\newcommand{\rank}{\mathrm{rank}}
\newcommand{\spann}{\mathrm{span}}

\newcommand{\sgn}{\mathrm{sign}}

\newcommand{\cov}{\mathrm{Cov}}

\newcommand{\Tr}{\mathrm{Tr}}
\newcommand{\vect}{\mathrm{Vec}}

\newcommand{\GMM}{\mathrm{GMM}}

\newcommand{\argmin}{\mbox{argmin}}
\newcommand{\argmax}{\mbox{argmax}}
\newcommand{\svd}{\mathrm{svd}}
\newcommand{\eigen}{\mathrm{eigen}}

\newcommand{\R}{\mathbb{R}}
\newcommand{\E}{\mathbb{E}}

\newcommand{\veps}{\varepsilon}

\renewcommand{\P}{\mathbb{P}}

\newcommand{\bff}{\mathrm{\bf f}}
\newcommand{\ba}{\mathrm{\bf a}}
\newcommand{\be}{\mathrm{\bf e}}

\newcommand{\bu}{\mathrm{\bf u}}
\newcommand{\bv}{\mathrm{\bf v}}

\newcommand{\bx}{\mathrm{\bf x}}
\newcommand{\by}{\mathrm{\bf y}}
\newcommand{\bz}{\mathrm{\bf z}}
\newcommand{\bA}{\mathrm{\bf A}}
\newcommand{\bB}{\mathrm{\bf B}}
\newcommand{\bZ}{\mathrm{\bf Z}}
\newcommand{\bC}{\mathrm{\bf C}}
\newcommand{\bD}{\mathrm{\bf D}}

\newcommand{\bF}{\mathrm{\bf F}}
\newcommand{\bK}{\mathrm{\bf K}}
\newcommand{\bT}{\mathrm{\bf T}}
\newcommand{\bW}{\mathrm{\bf W}}
\newcommand{\bG}{\mathrm{\bf G}}
\newcommand{\bM}{\mathrm{\bf M}}
\newcommand{\bH}{\mathrm{\bf H}}
\newcommand{\bI}{\mathrm{\bf I}}
\newcommand{\bg}{\mathrm{\bf g}}
\newcommand{\bP}{\mathrm{\bf P}}
\newcommand{\bV}{\mathrm{\bf V}}
\newcommand{\bQ}{\mathrm{\bf Q}}
\newcommand{\bR}{\mathrm{\bf R}}
\newcommand{\bS}{\mathrm{\bf S}}
\newcommand{\bU}{\mathrm{\bf U}}
\newcommand{\bX}{\mathrm{\bf X}}
\newcommand{\bY}{\mathrm{\bf Y}}

\newcommand{\xx}{\text{\boldmath $x$}}

\newcommand{\bm}{\text{\boldmath $m$}}

\newcommand{\bbeta}{\text{\boldmath $\beta$}}

\newcommand{\bDelta}{\text{\boldmath $\Delta$}}
\newcommand{\bone}{\mathrm{\bf 1}}
\newcommand{\bzero}{\mathrm{\bf 0}}

\newcommand{\bxi}{\text{\boldmath $\xi$}}

\newcommand{\bmu}{\text{\boldmath $\mu$}}
\newcommand{\bLambda}{\text{\boldmath $\Lambda$}}

\newcommand{\bTheta}{\text{\boldmath $\Theta$}}
\newcommand{\bSigma}{\text{\boldmath $\Sigma$}}

\newcommand{\bXi}{\text{\boldmath $\Xi$}}
\newcommand{\bPsi}{\text{\boldmath $\Psi$}}
\newcommand{\bPi}{\text{\boldmath $\Pi$}}
\newcommand{\bepsilon}{\text{\boldmath $\epsilon$}}

\newcommand{\cI}{\mathcal{I}}
\newcommand{\cS}{\mathcal{S}}

\newcommand{\cP}{\mathcal{P}}

\newcommand{\cW}{\mathcal{W}}

\ifNote

\theoremstyle{plain}
\newtheorem{thm}{Theorem}
\newtheorem{claim}{Claim}
\newtheorem{defn}{Definition}
\newtheorem{lem}{Lemma}
\newtheorem{cor}{Corollary}
\newtheorem{prop}{Proposition}
\newtheorem{ass}{Assumption}
\newtheorem{fact}{Fact}
\theoremstyle{definition}
\newtheorem{exm}{Example}
\newtheorem{rem}{Remark}

\else

\theoremstyle{plain}
\newtheorem{thm}{Theorem}[section]

\newtheorem{defn}{Definition}
\newtheorem{lem}{Lemma}
\newtheorem{cor}{Corollary}[section]
\newtheorem{prop}{Proposition}[section]
\newtheorem{ass}{Assumption}[section]

\theoremstyle{definition}

\fi

\RequirePackage[colorlinks,citecolor=blue,urlcolor=blue,linkcolor=blue]{hyperref}
\usepackage{color}

\newcommand{\aE}{\mathrm{AE}}

\renewcommand\footnotemark{}

\begin{document}

\def\spacingset#1{\renewcommand{\baselinestretch}%
{#1}\small\normalsize} \spacingset{1}


\title{Optimal Subspace Estimation Using Overidentifying Vectors via Generalized Method of Moments}


\author{Jianqing Fan$^*$
\thanks{Address: Department of ORFE, Sherrerd Hall, Princeton University, Princeton, NJ 08544, USA, e-mail: \textit{jqfan@princeton.edu}, \textit{yiqiaoz@princeton.edu}. The research was partially supported by NSF grants
DMS-1712591 and DMS-1662139 and NIH grant R01-GM072611.} and Yiqiao Zhong$^*$
\medskip\\{\normalsize $^*$Department of Operations Research and Financial Engineering,  Princeton University}
}
\date{}
\maketitle
\vspace*{-0.3 in}


\begin{abstract}
Many statistical models seek relationship between variables via subspaces of reduced dimensions. For instance, in factor models, variables are roughly distributed around a low dimensional subspace determined by the loading matrix; in mixed linear regression models, the coefficient vectors for different mixtures form a subspace that captures all regression functions; in multiple index models, the effect of covariates is summarized by the effective dimension reduction space.

Such subspaces are typically unknown, and good estimates are crucial for data visualization, dimension reduction, diagnostics and estimation of unknown parameters. Usually, we can estimate these subspaces by computing moments from data. Often, there are many ways to estimate a subspace, by using moments of different orders, transformed moments, etc. A natural question is: how can we combine all these moment conditions and achieve optimality for subspace estimation?

In this paper, we formulate our problem as estimation of an unknown subspace $\cS$ of dimension $r$, given a set of overidentifying vectors $\{ \bv_\ell \}_{\ell=1}^m$ (namely $m \ge r$) that satisfy $\E \bv_{\ell} \in \cS$ and have the form
$$
\bv_\ell = \frac{1}{n} \sum_{i=1}^n \bff_\ell(\xx_i, y_i),
$$
where data are i.i.d.\ and each function $\bff_\ell$ is known. By exploiting certain covariance information related to $\bv_\ell$, our estimator of $\cS$ uses an optimal weighting matrix and achieves the smallest asymptotic error, in terms of canonical angles. The analysis is based on the generalized method of moments that is tailored to our problem. Our method is applied to aforementioned models and distributed estimation of heterogeneous datasets, and may be potentially extended to analyze matrix completion, neural nets, among others.\end{abstract}

\noindent%
\vfill

\newpage
\sloppy

\section{Introduction}\label{sec:intro}

\subsection{Motivation of subspace estimation}\label{sec:motiv}

In statistics, many models are used to infer from data as simple relationship between variables as possible. Arguably, it is usually easier to conduct statistical analysis and interpret results if we find a simple relationship. For example, a family of well-studied statistical models is factor models:
\begin{equation}\label{model:factor}
\xx_i = \bB \bz_i + \bepsilon_i, \quad i \in [n] = \{1,2,\ldots,n\}.
\end{equation}
where $\xx_i \in \R^p$, $\bz_i \in \R^r$ and $r$ is usually (much) smaller than $p$. This model is useful since $\xx_i$ is characterized by a smaller number of variables $\bz_i$ (which are called factors) via a linear transformation $\bB$ (which is called the loading matrix), plus unexplained variable or noise $\bepsilon_i$. In particular, we have
\begin{equation*}
\xx_i - \bepsilon_i = \bB \bz_i \in \spann(\bB),
\end{equation*}
where $ \spann(\bB)$ is the linear span of column vectors of $\bB$. With this model, $\xx_i$ is roughly distributed around a $r$-dimensional subspace (if $\bB$ has full column rank), and the coordinates of $\xx_i$ on that subspace are determined by the factors $\bz_i$. Thus, $\spann(\bB)$ can be viewed as the intrinsic geometric characteristics of this model. Once this subspace is determined, the degree of freedom is reduced from $O(pr)$ to $O(r^2)$, and then statistical inference on $\bB$ and $\bz_i$ becomes easier. For an overview on dimensionality reduction, see \cite{Li18}.

Another family of models that receives much attention recently is mixture models. Despite having a long history, until recently theoretical analysis about initialization and convergence has been elusive \citep{AGHKT14, BalWaiYu17}. Consider a simple mixed linear model:
\begin{equation*}
y_i = \sum_{k=1}^K \bone\{ z_i = k \} ( \xx_i^T \bbeta_k + \epsilon_i),
\end{equation*}
where $\xx_i, y_i$ are observed and $z_i$, taking values in $\{1,2,\ldots, K\}$, is not observed. A special case is $K=1$, in which $z_i$ is a constant, and the model reduces to the linear regression model. The difficulty for general $K$ stems from the unobserved latent variable $z_i$. Although it is not easy to estimate and analyze each $\bbeta_k$, the subspace $\spann\{\bbeta_1,\ldots,\bbeta_K\}$ can be estimated fairly easily using the first and second moments, which is important for subsequent estimation of each $\bbeta_k$ \citep{YiCarSan16,  SedJanAna16}. More generally, multiple index models assume a semiparametric model
\begin{equation}\label{model:mim}
y_i = G(\xx_i^T \bbeta_1,\ldots, \xx_i^T \bbeta_K, \epsilon_i),
\end{equation}
where the form of $G$ is unknown. In all these examples, the subspace
\begin{equation*}
\cS = \spann(\bB) \quad \text{or} \quad \cS = \spann\{\bbeta_1,\ldots,\bbeta_K\}
\end{equation*}
plays a pivotal role, since it captures and summarizes the information of one part of the variables (often covariates) in relation to the other. This motivates us to consider the problem of estimating $\cS$ alone in a general setting, which we call \textit{subspace estimation} in this paper.

The advantage of studying this problem is three-fold. (1) The subspace $\cS$ is intrinsic geometrically, which is invariant to rotation, and therefore there is no identifiability issue for subspace estimation. (2) After obtaining a good estimate of $\cS$, it is easier for statisticians to visualize data, to reduce data dimensions, to estimate unknown parameters (with a largely reduced degree of freedom), and to study model diagnostics. In this aspect, subspace estimation can be viewed as an intermediate estimation problem. (3) It is less difficult to estimate a subspace than unknown parameters, because instead of finding the maximum likelihood estimator (MLE) or running an EM algorithm, we need only moment conditions from the data, which are much easier to compute.

Relevant to the third point, often we find ourselves in situations where we have many ways to estimate $\cS$. For example, in factor models, the covariance matrix is usually used to determined $\spann(\bB)$; however, if $\xx_i$ has nonzero means, the non-centered form is also informative since $\E \xx_i \in \spann(\bB)$. This observation has led to a recent work by \cite{LetPel17}. And in multiple index models, it is known that transforming the moments can be helpful \citep{Li92}.
These considerations lead naturally to our problem formulation.

\subsection{Problem formulation}\label{sec:formulate}

Suppose we have data $\{(\xx_i, y_i)\}_{i=1}^n$, where $y_i$ is the response variable, $\xx_i$ is the vector of covariates (or predictors), and $n$ is the sample size. A special case is that all $y_i$ are set to a constant, or equivalently we omit $y_i$ altogether (a.k.a.\ unsupervised learning). Assume that $(\xx_i, y_i)$ is i.i.d., with an unknown distribution $P$. Our goal is to estimate an unknown linear subspace $\cS = \cS(P)$ of $\R^p$ with dimension $\dim(\cS) = r \ge 1$, where $r$ is fixed and unknown, from a set of overidentifying conditions:
\begin{align}
& \bv_\ell = \frac{1}{n} \sum_{i=1}^n \bff_\ell(\xx_i, y_i) \in \R^p, \qquad \ell \in [m], \label{def:v1} \\
&\text{with} ~~ \E \bv_\ell \in \cS, ~ \forall \, \ell \in [m], \label{def:v2}
\end{align}
where $m \ge r$, and $\bff_1(\xx_i, y_i), \ldots, \bff_m(\xx_i, y_i)$ are known functions with finite second moments. Often, the vector $\bv_\ell$ is the empirical moments of $\xx_i$ and $y_i$, but our definition here is very general, as $\bff_\ell$ is a generic function. In general, we have more than enough conditions to determine $\cS$, since $m \ge r$ and the linear span of $\E \bv_1, \ldots, \E\bv_m$ is exactly $\cS$ except for degenerate cases. Thus, it is reasonable to expect a good estimator $\hat \cS$ from the statistics $\bv_1, \ldots, \bv_m$. The question, then, is how to produce an estimator in an \textit{optimal} way.

Note that we do not make assumptions on $\xx_i$ or $y_i$ directly, other than the i.i.d.\ assumption. Assumptions will be made, and results will be stated, in terms of $\bff_\ell(\xx_i, y_i)$. Also note that $\bv_\ell$ is not required to have unit $\ell_2$ norm, since it is implicitly rescaled in our procedure (see Section~\ref{sec:procedure}).

\subsection{Combining overidentifying vectors optimally}\label{sec:combine}

\begin{figure*}[b!]
    \centering
    \includegraphics[scale=0.45]{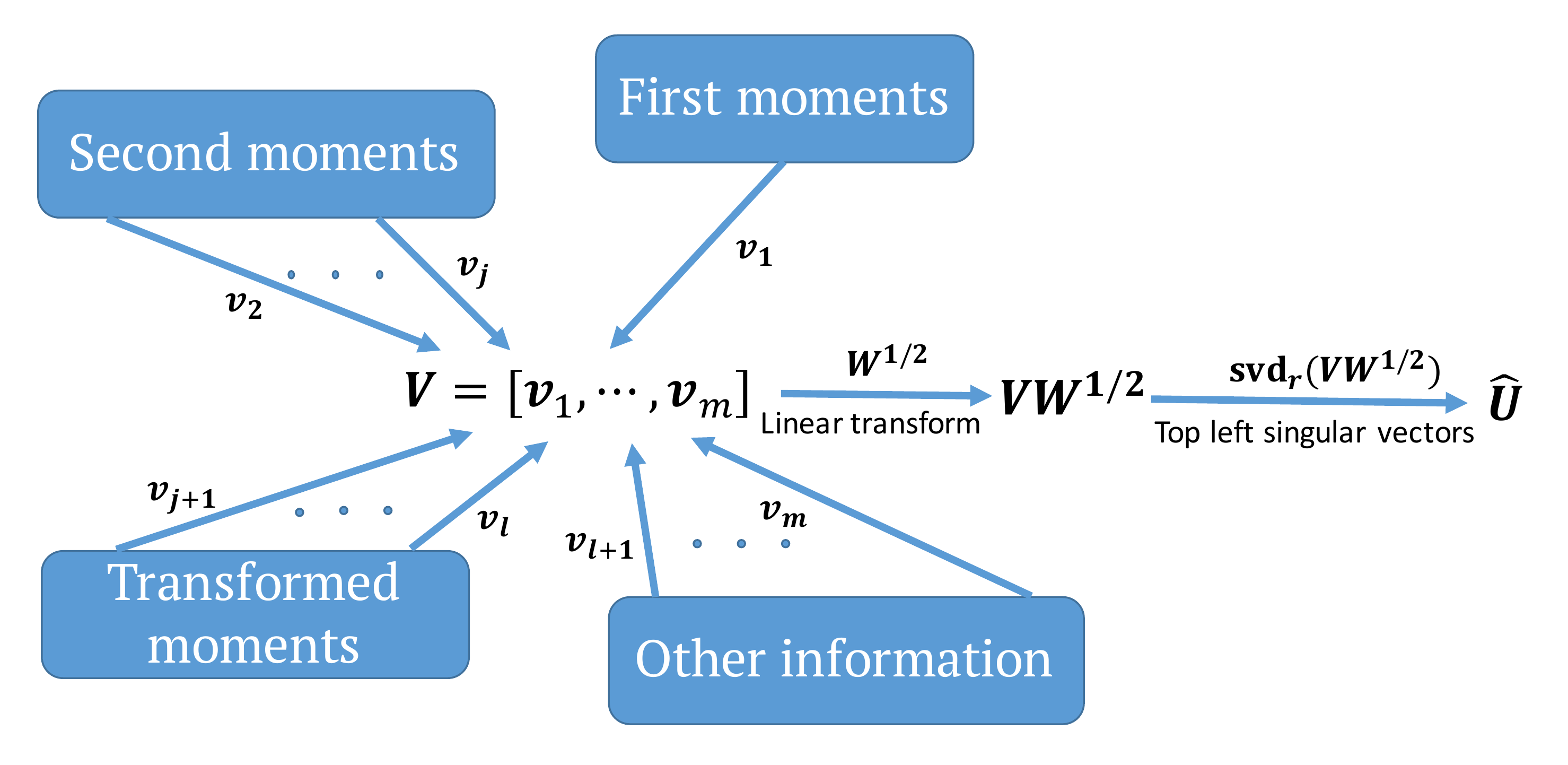}
    \caption{A diagram showing the procedure}\label{fig:diagram2}
\end{figure*}

Let us approach this problem by first assuming $r$ is known; otherwise, it can be consistently estimated (see Section~\ref{sec:dimest}). In general, to determine a subspace of dimension $r$, we need $r$ linearly independent vectors. Now given a possibly large pool of $\bv_\ell$, we have to look for the common and dominant space. A natural way is to consider singular vectors. Let $\svd_r(\cdot)$, and respectively $\eigen_r(\cdot)$, denote top $r$ left singular vectors and eigenvectors of a matrix. Suppose we compute
\begin{equation*}
\svd_r \left([\bv_1,\ldots,\bv_m] \right) \quad \text{or equivalently} \quad \eigen_r \left( \sum_{\ell=1}^m \bv_\ell \bv_\ell^T \right).
\end{equation*}
Then the resulting $r$ singular (or eigen-)vectors make use of all given moment conditions. Had these $\bv_\ell$ been i.i.d., this method would be the same as computing the principal components. However, we do not have the luxury to make such assumptions. More often than not, they have different variances, and they may be arbitrarily dependent. To take it into account, we modify our method with a symmetric weighting matrix $\bW = (w_{j\ell}) \in \R^{m \times m}$, and compute
\begin{equation*}
\svd_r \left(  [\bv_1,\ldots,\bv_m] \bW^{1/2} \right) \quad \text{ equivalently} \quad \eigen_r \left( \sum_{j, \ell=1}^m w_{j\ell} \bv_j \bv_\ell^T \right) = \eigen_r \big (\bV \bW \bV^T \big ),
\end{equation*}
where $\bW^{1/2} \in \R^{m \times m}$ is such that its square equals $\bW$. This method is the same as first applying a linear transformation $\bW^{1/2}$ to $[\bv_1,\ldots,\bv_m]$ and then implementing the unweighted method. For any fixed $\bW$, the transformation results in $m$ new vectors that satisfy \eqref{def:v2}. Our task, therefore, is to find an optimal $\bW$.

As we will show, in the sense of large sample asymptotic theory, the optimal choice of $\bW$ is the inverse of certain pseudo-covariance matrix of these vectors (or equivalently $\bff_\ell$)---see \eqref{def:sigma0}. Intuitively, we scale down vectors with large variances, and scale them up otherwise. Our optimality results for such $\bW$ include, for example, the smallest asymptotic expectation of
\begin{equation*}
d(\hat \cS, \cS)^2 = \| \cP_{\hat \cS} - \cP_{\cS} \|_F^2,
\end{equation*}
where $ \cP_{\hat \cS}, \cP_{\cS}$ are projections to a subspace. This criterion is similar to the mean squared error, and can be also expressed in terms of canonical angles---see \eqref{eq:mse}.

Our analysis is based on, but not directly derived from, the generalized method of moments (GMM). Although PCA for subspace estimation has a long history and a well-established theory \citep{Pearson01, Hotelling33, Anderson63}, overidentifying vectors are much less studied. General optimality results of GMM are developed in the seminal paper by \cite{Hansen82}, which, however, does not directly apply to our problem. While GMM is widely studied both theoretically and empirically, previous works mostly tackle problems where the MLE of unknown parameters is not available or very difficult to compute. Our paper shows that GMM with a suitable form of weighting matrices produces a simple closed-form estimator, which is useful as an intermediate estimator.

We shall call our proposed method that uses the optimal $\bW$ the \textit{GMM subspace estimator}.

\subsection{Related works and paper organization}\label{sec:related}

The paper is organized as follows. In Section~\ref{sec:subspGMM}, we propose our estimation procedure whose applications are elucidated in Section~\ref{sec:ex}. In Section~\ref{sec:theory}, we present our theoretical analysis and show optimality of our procedure under nearly minimal assumptions.  Numerical simulations and a dataset example are presented in Section~\ref{sec:sim} and~\ref{sec:real} to support our theory. Finally, in Section~\ref{sec:disc}, we discuss possible extensions.

\section{Subspace GMM estimator}\label{sec:subspGMM}


We derive our subspace estimator in a way similar to the classical GMM theory. However, there are several departures: (1) matrix representation of subspaces involves identification of matrices up to rotation, since inherently the parameter space is a Grassmann manifold; (2) a specific blockwise form of the weighting matrix allows simple and fast computation, while enjoying optimality results---see Sections~\ref{sec:fast} and~\ref{sec:optproc}; (3) for this particular estimation problem, we develop clear methods and analysis for subtle issues, such as singularity of weighting matrices---see Section~\ref{sec:singular}.

\subsection{Estimation via GMM framework}

Before formally deriving the GMM estimator, we first address the representation of the subspace $\cS$. Let $O(p,r)$ be the set of $p\times r$ matrices consisting of orthonormal column vectors. Every $\cS$ is associated with some $\bU^* \in O(p,r)$ whose columns lie in $\cS$, that is,
\begin{equation}\label{def:Ustar}
\bU^* = [ \bu_1^*, \bu_2^*, \ldots, \bu_r^*] \in \R^{p \times r}, \quad \text{with} ~ (\bU^*)^T \bU^* = \bI_r, ~ \text{and} ~ \bu_j^* \in \cS, ~ \forall\, j \in [r].
\end{equation}
Such matrix representation $\bU^*$ is unique up to a rotation: a $p \times r$ matrix satisfies the conditions in \eqref{def:Ustar} if and only if it has the form $\bU^* \bR$, where $\bR \in O(r)$ is an orthogonal matrix of size $r$. This is because any two sets of orthonormal bases can be mapped to each other by an orthogonal matrix. It is clear, then, that the projection matrix $\bP_{\bU^*}^\bot := \bI_p - \bU^* (\bU^*)^T \in \R^{p \times p}$ is unique.

Therefore, we can represent the set of all possible $\cS$ as $O(p,r)/O(r)$, that is, the space $O(p,r)$ up to rotation. In differential geometry, the space $O(p,r)$ is called the Stiefel manifold, and $O(p,r)/O(r)$ is called the Grassmann manifold \citep{EdeAriSmi98}.

Using the representation of $\cS$, we can rewrite \eqref{def:v2} into the following estimating equations, which are usually called the \textit{population moment condition} in the GMM literature:
\begin{equation*}
\E (\bI_p - \bU^* (\bU^*)^T) \bv_\ell = \bzero, \qquad \forall\, \ell \in [m].
\end{equation*}
Now it is natural to introduce the GMM estimator as follows. For any matrix $\bU^* \in O(p,r)$ with orthonormal columns, we concatenate all vectors $(\bI_p - \bU \bU^T) \bv_\ell$ into a single vector:
\begin{equation}\label{def:gu}
\bg(\bU):= \left(  \begin{array}{c}
(\bI_p - \bU \bU^T) \bv_1 \\ \vdots \\ (\bI_p - \bU \bU^T) \bv_m
\end{array} \right)  \in \R^{\bar m}, \qquad \text{where} ~ \bar m := mp.
\end{equation}
We also denote $\bar p = rp$. Thus, the function $\bg$ is a nonlinear map from $\R^{\bar p}$ to $\R^{\bar m}$. For any positive definite matrix $\bW = (w_{k \ell} )_{k,\ell \in [m]} \succ \bzero$ in $\R^{m \times m}$, we define the weighting matrix as
\begin{equation}\label{def:Wbar}
\bar \bW = \bW \otimes \bI_p = \left( \begin{array}{ccc}
w_{11} \bI_p & \cdots & w_{1p} \bI_p \\
\vdots & \ddots & \vdots \\
w_{1m} \bI_p & \cdots & w_{mm} \bI_p
\end{array} \right) \in \R^{\bar m \times \bar m},
\end{equation}
where $\otimes$ is the Kronecker product. By construction, $\bar \bW$ is also a definite positive matrix (see Lemma~\ref{lem:Kron}). Here, both $\bW$ and $\bar \bW$ can be random. Following the classical GMM approach \citep{Hansen82}, we define the GMM estimator $\hat \bU$ as a minimizer of $Q(\bU)$, which is quadratic in $\bg(\bU)$ and thus quartic (i.e., involving fourth moments) in $\bU$.
\begin{equation}\label{def:Q}
\hat \bU \in \argmin_{ \bU \in O(p,r)} Q(\bU) \qquad \text{where} ~ Q(\bU) := [ \bg(\bU)]^T \bar \bW \bg(\bU) .
\end{equation}
A few remarks are in order. First, as a minimizer of \eqref{def:Q}, $\hat \bU$ is not uniquely defined. It is clear that $\hat \bU$ is a minimizer of \eqref{def:Q} if and only if $\hat \bU \bR$ is also a minimizer for any $\bR \in O(r)$. An alternative way that would circumvent this issue is to define $Q$ as a function of $\bU \bU^T$, since a minimizer of this function does not depend on the choice of $\bR$. However, for the ease of expositions and analysis, we focus on the current form defined in \eqref{def:Q}. Second, in the minimization problem, the parameter space $\{\bU: \bU \in O(p,r)\}$ is a Stiefel manifold in $\R^{pr}$, which has dimension $pr - r(r+1)/2$; whereas, the standard GMM framework usually assumes the parameter space contains some neighborhood (or ball) around $\bU^*$ \citep{Hall05}. Third, relevant to the first two remarks, one may consider defining $\hat \bU$ as a minimizer of $Q$ over the Grassmann manifold, which also resolved the non-uniqueness issue in the first remark. However, we avoid such treatment here due to heavy machinery from differential manifolds.

\subsection{Computing $\hat \bU$ via eigendecompostion}\label{sec:fast}

With the block matrix form of $\bar \bW$, we are able to simplify the optimization problem and compute $\hat \bU$ via the standard eigen-decomposition computation.

We observe that although the objective function in \eqref{def:Q} has a quartic form, it is equivalent to a quadratic function, and consequently, the optimization problem \eqref{def:Q} can be solved very efficiently. This is due to the fact that $\bI_p - \bU \bU^T$ is a projection matrix, so $(\bI_p - \bU \bU^T)^2 = \bI_p - \bU \bU^T$, and
\begin{align*}
[ \bg(\bU)]^T \bar \bW \bg(\bU) &= \sum_{k,\ell=1}^m w_{k \ell} \bv_k^T (\bI_p - \bU \bU^T)\bv_\ell =  \sum_{k,\ell=1}^m w_{k \ell} \bv_k^T\bv_\ell  - \sum_{k,\ell=1}^m w_{k \ell} \Tr(\bU^T \bv_\ell \bv_k^T \bU) \\
&= \sum_{k,\ell=1}^m w_{k \ell} \bv_k^T\bv_\ell  - \Tr \Big( \bU^T\sum_{k,\ell=1}^m w_{k \ell} \bv_\ell \bv_k^T \, \bU \Big),
\end{align*}
where $\Tr(\cdot)$ is the trace of a matrix. Hence, we obtain the following result.
\begin{prop}\label{prop:eigen}
Solving the optimization problem \eqref{def:Q} is equivalent to solving
\begin{equation}\label{opt:eigen}
\max_{\bU \in O(p,r)} \Tr(\bU^T \bV \bW \bV^T \bU), \qquad \text{where}~ \bV := [\bv_1,\ldots, \bv_m] \in \R^{p \times m}.
\end{equation}
The columns of its solution $\hat \bU$ are given by the top $r$ eigenvectors of $\bV \bW \bV^T$.
\end{prop}
This proposition provides a way of computing $\hat \bU$ given $\bW$. Here `top eigenvectors' refer to those eigenvectors with largest eigenvalues. Under a nondegeneracy assumption (Assumptions~\ref{ass:span}), $\bV \bW \bV^T$ has a non-vanishing gap between its $r$th and $(r+1)$-th largest eigenvalues for large $n$, and thus $\hat \bU$ is unique up to rotation. (A large eigen-gap also ensures numerical stability.) See details in Section~\ref{sec:dimest}.

We remark that the above simplification hinges on the block matrix form of $\bar \bW$. For a general $\bar m \times \bar m$ weighting matrix $\bar \bW$, there is no simple way to solve \eqref{def:Q}, which is a genuine quartic function in $\bU$.

\subsection{Two-step estimation procedure}\label{sec:procedure}

With an appropriate choice of the weighting matrix, the GMM produces, in general, asymptotically efficient estimators. This can be usually achieved through a two-step estimation procedure: the first step is obtaining a consistent estimator, which is used to compute an optimal weighting matrix $\bW$; and the second step is to solve an optimization problem with the $\bW$ computed from the first step. The first step is often easy, and in particular it is so for our problem, since eigen-decomposition of $\bV \bW \bV^T$ with any $\bW \succ \bzero$ leads to a consistent estimator. In the second step, we choose $\bW$ in a way such that it converges in probability (as $n \to \infty$) to $(\bSigma^*)^{-1}$, where $\bSigma^* = (\Sigma_{j\ell}^*) \in \R^{m \times m}$ is defined as
\begin{equation}\label{def:sigma0}
\Sigma_{j\ell}^* = \E \left[   \bff_j(\xx_i, y_i)^T \left(\bI_p - \bU^* (\bU^*)^T \right)  \bff_{\ell}(\xx_i, y_i)     \right], \quad \forall\, j, \ell \in [m].
\end{equation}
Note that this definition does not depend on $i$. In the matrix form, it is the same as
\begin{equation*}
\bSigma^* = \E \left[ \bF_i^T  \left(\bI_p - \bU^* (\bU^*)^T \right) \bF_i \right], \quad \text{where} ~ \bF_i = [ \bff_1(\xx_i, y_i), \ldots,  \bff_m(\xx_i, y_i)] \in \R^{p \times m}.
\end{equation*}
The particular structure of $\bSigma^*$ is closely related to the covariance matrix of $\bff_1,\ldots,\bff_m$, and it comes with optimality guarantees---see Section~\ref{sec:normal} and~\ref{sec:optproc}. Given an initial consistent $\hat \bU^0$, we can find a consistent estimator of $\bSigma^*$, denoted by $\hat \bSigma$, using the natural plug-in estimator, that is, the sample mean of  $ \bff_j(\xx_i,y_i)^T \big(\bI_p - \hat \bU^0 (\hat \bU^0)^T \big)  \bff_{\ell}(\xx_i,y_i)$ for all $j,\ell$, namely,
\begin{equation}\label{fan1}
    \hat \bSigma = n^{-1} \sum_{i=1}^n \bF_i^T  \left(\bI_p - \hat \bU^0 (\hat \bU^0)^T  \right) \bF_i.
\end{equation}
 Then, setting $\bW$ to be $(\hat \bSigma)^{-1}$ and solving the eigen-decomposition \eqref{prop:eigen} again, we obtain the final GMM estimator, denoted by $\hat \bU^{\GMM}$. Our estimation procedure is formally described as follows:

\begin{itemize}

\item[1.] Obtain an initial consistent estimator $\hat \bU^0$. For example, one can choose an initial weighting matrix $\bW = \bI_p$ in \eqref{opt:eigen} and compute top $r$ eigenvectors of $\bV \bV^T$.

\item[2.] For each $j, \ell \in [m]$ with $j \le \ell$, calculate
\begin{equation}\label{def:sigmahat}
\hat \Sigma_{j \ell} = \frac{1}{n} \sum_{i=1}^n \left[ \bff_j(\xx_i, y_i)^T \left(\bI_p - \hat \bU^0 (\hat \bU^0)^T \right)  \bff_{\ell}(\xx_i, y_i) \right]
\end{equation}
and set $\hat \Sigma_{\ell j} = \hat \Sigma_{j \ell}$. Form a matrix $\hat \bSigma = (\hat \Sigma_{j\ell}) \in \R^{m \times m}$, which is equivalent to computing \eqref{fan1} in matrix form.

\item[3.] If $r = \dim(\cS)$ is not known in advance, estimate $r$ as suggested in Section~\ref{sec:dimest}.

\item[4.] Set $\bW$ to be $(\hat \bSigma)^{-1}$ or \eqref{def:generalinv} below, and then compute the top $r$ eigenvectors of $\bV \bW \bV^T$. Obtain the estimator $\hat \bU^{\GMM} \in \R^{p \times r}$ by combining these eigenvectors into a matrix.

\end{itemize}

We make a fews remarks about the above estimation procedure. First, in principle there are many ways to produce a consistent $\hat \bU^0$. One may choose a subset of vectors from $\bv_1,\ldots, \bv_m$, if it is known that such subset spans the target subspace $\cS$. Often, standard (or vanilla) estimators serve as good initial estimates. An alternative approach, as often used in the GMM literature, is to conduct iterative GMM, which is to estimate a sequence of $\hat \bU^t$, where $\hat \bU^t$ is calculated based on $\hat \bU^{t-1}$.

Second, in cases where $r$ is not known, one may follow Section~\ref{sec:dimest} to estimate $r$, and the method only involves the eigenvalues of $\bV \bW \bV^T$. Thus, Step 3 does not require additional computational cost.

Third, if $\hat \bSigma$ is singular (namely, not invertible) or nearly singular, the following variant is preferred over $(\hat \bSigma)^{-1}$. Given a parameter $\delta_n \ge 0$, compute the eigen-decomposition of $\hat \bSigma = \bar \bU\, \diag\{ \bar \lambda_1,\ldots, \bar \lambda_m\} \bar \bU^T$, and set
\begin{equation}\label{def:generalinv}
\bW = \bar \bU \, \diag\{ \psi(\bar \lambda_1),\ldots ,\psi(\bar \lambda_m) \} \bar \bU^T , \quad \text{where } \psi(x):= x^{-1} \bone\{ x > \delta_n \}.
\end{equation}

It generalizes the usual matrix inverse: if $\delta_n = 0$, then \eqref{def:generalinv} gives the Moore-Penrose pseudoinverse $(\hat \bSigma)^{+}$; and if $\delta_n > 0$, then \eqref{def:generalinv} applies a de-noising step with a threshold $\delta_n$ before taking the Moore-Penrose pseudoinverse. In this regard, \eqref{def:generalinv} can be viewed as computing the pseudoinverse with a hard-thresholding operation. We suggest choosing $\delta_n$ such that $\delta_n = o(1)$ and $\sqrt{n}\, \delta_n \to \infty$. See the analysis in Section~\ref{sec:singular},

As shown in Section~\ref{sec:singular}, the presence of redundant vectors is often the cause of singularity of $\hat \bSigma$, in which case using the generalized inverse \eqref{def:generalinv} still gives optimality guarantee with a suitable $\delta_n$.

\subsection{Extensions}\label{sec:extension}

So far, we have considered combining information that has an average form as required by \eqref{def:v2}. Typically, this is useful when such $\bv_\ell$ is derived from the method of moments. However, our framework can be used to accommodate the information of $\cS$ which does not have the average form. For example, in the factor model, any single $\xx_i$ satisfies $\E \xx_i = \bB \E \bz_i \in \cS$, but it does not admit the average form. This kind of individual information can also be incorporated in our GMM framework \eqref{opt:eigen} under certain structure.

In order to apply our asymptotic theory, we assume that we have individual moment information $\bm_i \in \cS$ based on the $i$th data.  For the factor model, we naturally take $\bm_i = \xx_i$.  These individual moments are naturally aggregated as the matrix $\bM = n^{-1} \sum_{i=1}^n \bm_i \bm_i^T$. Under additional structure that is similar to the factors, we may assume that the eigenspace of $\E \bM - \bSigma_0$ falls in $\cS$ for some given symmetric $\bSigma_0$.  We can now take the moment conditions
$$
  \bv_\ell = (\bM-\bSigma_0) \be_\ell = n^{-1} \sum_{i=1}^n (\bm_i \bm_i^T \be_\ell - \bSigma_0 \be_\ell),
$$
which has the form of averages.  Therefore, it falls in our framework.  The unweighted aggregation of these moments in Proposition~\ref{prop:eigen} leads to the principal component analysis of the matrix
$$
    \sum_{\ell=1}^n \bv_\ell \bv_\ell^T = (\bM-\bSigma_0)^2.
$$
We remark that the contribution of the matrix $\bSigma_0$ is negligible under the pervasive conditions \citep{bai2003inferential, fan2013large}.  Hence, one may simply use the eigenspace spanned the matrix $\bM$ to obtain reasonable estimates, while our two-step method allows for more efficient construction.

More generally,  suppose that we are given a symmetric matrix $\bM \in \R^{p \times p}$ that is computed from the data such that $\spann(\eigen_r(\bM))$ is a consistent estimator of $\cS$. A natural way to combine $\bM$ with our GMM estimator is to compute
\begin{equation*}
\eigen_r(\kappa \bM + \bV \bW \bV^T)
\end{equation*}
where $\kappa$ is a suitable parameter. The hope is that, a good choice of $\kappa$ may lead to a more efficient estimator.

While it seems difficult to determine an optimal parameter $\kappa$ in general, we identity cases in Section~\ref{sec:augment}, in which introducing $\kappa \bM$ does not bring extra asymptotic efficiency; or in other words, the optimal $\kappa = 0$. In practice, one may consider a few choices of values for $\kappa$, and use the bootstrap to determine the best $\kappa$.

\section{Examples}\label{sec:ex}
We discuss four typical applications to exemplify the general procedure proposed in the previous section.

\subsection{Factor models}\label{sec:factormodel}
First, let us consider the simple factor model \eqref{model:factor},
where, $\xx_i, \bepsilon_i \in \R^p$ and $\bz_i \in \R^r$ are random vectors, and $\bB \in \R^{p \times r}$ is a fixed and unknown matrix. Only $\xx_i$ is observed, and $\bz_i$ and $\bepsilon_i$ represent, respectively, latent factors and noise. Suppose $\{ \bz_i, \bepsilon_i \}_{i=1}^n$ are i.i.d., and all vectors are jointly independent. Let $\bmu_{\bz} = \E(\bz_i)$ and $\bSigma_{\bz} = \cov(\bz_i)$. We also assume that $\E \bepsilon_i = \bzero$, and $\cov(\bepsilon_i) = \sigma^2 \bI_p$, where $\sigma$ is known.

This model and its variants are well studied, and have wide applications in econometrics, psychology, etc. Usually, $r$ is much smaller than $p$. Under this model, the $p$ dimensional vector $\xx_i$ is predominately determined by a linear combination of a small number of `factors' $\bz_i$. The covariance of $\bepsilon_i$ is assumed to be simple in this paper, while more sophisticated and general structures have been considered \citep{ConKor93,  forni2000generalized, bai2003inferential, fan2013large, bai2013principal, fan2016projected, forni2017dynamic}. Nevertheless, our simple factor model retains essential features, as dimension reduction via principal component analysis is routinely employed in the literature.

Let $\cS = \spann(\bB)$ be the target subspace we want to estimate. It is the left singular subspace of $\bB$, and is also the subspace spanned by principal component directions. Once we have a good estimate of $\cS$, it is much easier to estimate $\bB$, since the degree of freedom reduces from $pr$ to $r^2$. A nice property of $\cS$ is that there is no identifiability issue as is common in factor models, since $\spann(\bB \bQ) = \spann(\bB)$ holds for any invertible matrix $\bQ$.

Routinely, eigen-decomposition of the empirical covariance matrix of $\xx_i$, namely $n^{-1}\sum_{i=1}^n(\xx_i - \bar \xx)(\xx_i - \bar \xx)^T$ where $\bar \xx = n^{-1} \sum_{i=1}^n \xx_i$, forms the cornerstone of estimation of $\cS$. If $\bmu_{\bz} \neq \bzero$, however, the first moments supply complementary information to the covariance matrix (second moments). To see this, notice
\begin{equation*}
\E \xx_i = \bB \bmu_{\bz} \in \cS.
\end{equation*}
Intuitively, if $\bmu_{\bz}$ is very large, then $\bar \xx$ is useful to estimate $\cS$. This motivates combining both the first and second moments. To this end, we set
\begin{align}\label{eq:vfactor}
\bv_1 = \frac{1}{n} \sum_{i=1}^n \xx_i, \quad \bv_{1+\ell} = \frac{1}{n} \sum_{i=1}^n \xx_i \xx_i ^T \be_\ell - \sigma^2 \be_\ell, \quad \forall\, \ell \in [p],
\end{align}
where $\{ \be_\ell\}$ is the standard basis. Here $\bv_{1+\ell}$ is simply a projection of the matrix $n^{-1}\sum_i\xx_i\xx_i^T - \sigma^2\bI_p$ to each basis vector $\be_\ell$. We can then follow the procedure outlined in Section~\ref{sec:procedure}.

Note that, we can replace $\{ \be_\ell\}$ by any set of $p$ linearly independent vectors, and the same optimality guarantee holds. This is because a change of basis only results in a linear transformation of moments, and our theory ensures optimality among all linear transformations (Section~\ref{sec:optproc}). Moreover, replacing $\{ \be_\ell\}$ by a set of overcomplete vectors does not give a better estimate, since projections onto overcomplete vectors only provide redundant information (Section~\ref{sec:singular}).

We also note that in \eqref{eq:vfactor}, we assume $\sigma$ is known, so we can subtract $\sigma^2 \be_\ell$ to remove the effect from the noise term. This ensures $\E \bv_{1+\ell} \in \cS$ and thus our theory is applicable. In the case of an unknown $\sigma$, one may consider estimating $\sigma$ first, or updating estimates iteratively.

Finally, we remark that recently, \cite{LetPel17} also utilize first moments, along with second moments, to achieve improved estimation in factor models. Both theoretical and empirical evidence are shown to justify the use of first moments. In Section~\ref{sec:augment}, we will discuss the connection to this work.


\subsection{Mixture models}\label{sec:mixturemodel}

Now let us consider mixtures of generalized linear models (GLM). Suppose we have i.i.d.\ data $\xx_i$ and $y_i$, with $\xx_i \in \R^p$, $y_i \in \R$ and $i\in [n]$. An unobserved variable $z_i \in \{1,2,\ldots,K\}$, independent of $\xx_i$ and $y_i$, indicates which model $\xx_i$ and $y_i$ are generated from. To be precise, each $z_i$ is a multinomial variable with $\P(z_i = k) = \pi_k$, where $\pi_k$ is a parameter; and conditioning on $z_i$ and $\xx_i$, $y_i$ has a density function (or probability mass function):
\begin{equation}\label{model:mixmodel}
f(y|\xx_i, z_i=k; \bTheta) = h_k(y; \theta_k(\xx_i)), \quad \text{where} ~\theta_k(\xx_i) = \xx_i^T \bbeta_k + \beta_{k0}.
\end{equation}
Here, $\bbeta_k \in \R^p$ and $\beta_{k0} \in \R$ are the unknown parameters of the $k$th GLM model. The parameter space $\bTheta$ is the set of all these parameters. For each $i \in [n]$, depending on the value $z_i$ takes, the response variable $y_i$ is generated from one of the $K$ GLMs. Special cases include:
\begin{itemize}
\item {mixed linear regression}
\begin{equation}\label{model:linreg}
y_i = \sum_{k=1}^K \bone\{ z_i = k \} ( \xx_i^T \bbeta_k + \beta_{k0} + \sigma_k \epsilon_i), \quad \text{where } \sigma_k \ge 0;
\end{equation}
\item {mixed logistic regression}
\begin{equation}\label{model:logreg}
\P(y_i=1 | \xx_i, z_i = k) = \varphi( \xx_i^T \bbeta_k +  \beta_{k0}), \quad \text{where } \phi(t) = \frac{1}{1 + e^{-t}}.
\end{equation}
\end{itemize}

Mixtures models and a broader family of latent variable models can be tackled by EM algorithms \citep{DemLaiRub77}. While local convergence is established under various conditions \citep{Wu83,XuJor96,BalWaiYu17}, EM algorithms are usually susceptible to local minima \citep{Jin16}. In recent years, tensor-based methods, which utilize moments of $\xx_i$ (up to the third order moments), are proved to produce consistent estimators under some conditions on the covariates $\xx_i$ \citep{AGHKT14,YiCarSan16}. Usually, the covariates are required to be i.i.d.\ normal, but moderate extensions are possible. A crucial step of these works is to seek a whitening matrix, which is then used to construct tensors that have a special orthogonal structure. The columns of such whitening matrix have the same linear span as $\spann\{ \bbeta_1,\ldots, \bbeta_K \} $.

Let $\cS = \spann\{ \bbeta_1,\ldots, \bbeta_K \} $ be our target subspace. Usually $K$ is much smaller than $p$, so a good estimate of $\cS$ is helpful for the estimation of $\bbeta_k$. In the aforementioned papers, typically one can estimate $\cS$ through second moments of $\xx_i$, assuming $\xx_i \sim N(\bzero,\bI_p)$. For example, in the case of mixed logistic regression, one can use Stein's identity and derive
\begin{equation*}
\E \left[ y_i (\xx_i \xx_i^T - \bI_p)	\right] = \sum_{k=1}^K \pi_k \bbeta_k \bbeta_k^T \E [ \varphi''( \xx_i^T \bbeta_k +  \beta_{k0})].
\end{equation*}
Note that each column of $\E \left[ y_i (\xx_i \xx_i^T - \bI_p)	\right]$ lies in $\cS$, so it leads to an obvious construction of moments that fit in our framework \eqref{def:v1}--\eqref{def:v2}:
\begin{equation}\label{eq:vell}
\bv_{\ell} = \frac{1}{n} \sum_{i=1}^n y_i (\xx_i \xx_i^T \be_\ell - \be_\ell), \quad \forall\, \ell \in [p].
\end{equation}
This does not work for mixed linear regression, since $\E \left[ y_i (\xx_i \xx_i^T - \bI_p) \right]$ would vanish. A better choice for mixed linear regression is $\E \left[ y_i^2 (\xx_i \xx_i^T - \bI_p)	\right]$, due to
\begin{equation*}
\E \left[ y_i^2 (\xx_i \xx_i^T - \bI_p)	\right] = 2\sum_{k=1}^K \pi_k \bbeta_k \bbeta_k^T.
\end{equation*}
This leads to a similar construction as \eqref{eq:vell}, except that $y_i$ is replaced by $y_i^2$.

Apart from these second moments proposed in the literature, one may consider first moments, or moments with $y_i$ transformed by a nonlinear function. For instance, one may construct
\begin{align}\label{eq:vell2}
\bv_\ell = \frac{1}{n} \sum_{i=1}^n h_{\ell}(y_i) \xx_i, \quad \text{or} \quad \bv_{\ell,j} = \frac{1}{n} \sum_{i=1}^n h_{\ell}(y_i) ( \xx_i \xx_i^T \be_j - \be_j),
\end{align}
where, for both expressions, one may choose a function $h_\ell(y_i) = \cos(y_i/t_\ell + \gamma_\ell)$ for any $t_\ell>0$ and $\gamma_\ell \in \R$. One can also simply choose $h_\ell(y_i) = y_i$, leading to the (untransformed) first moments of $\xx_i$. By Stein's identity, they all satisfy the condition that $\E \bv_\ell$ (or $\E \bv_{\ell,j}$) lies in $\cS$. There are other useful construction in the literature, e.g., \cite{SunIoaMon14}.

After constructing these vectors $\bv_\ell$, we can use our estimation procedure to combine these vectors for optimal asymptotic efficiency, as promised by our theory. In general, the optimality result holds as long as mild regularity conditions (Assumption~\ref{ass:span} and~\ref{ass:redundancy}) are satisfied.

\subsection{Multiple index models}\label{sec:indexmodel}

The mixture model \eqref{model:mixmodel} discussed above can be subsumed in multiple index models, which are semiparametric models of the form \eqref{model:mim}, where $\bbeta_1, \ldots, \bbeta_K \in \R^p$ are unknown parameters, $\xx_i \in \R^p, y_i \in \R$ are i.i.d.\ data, $\epsilon_i$ is unobserved and i.i.d., and the function $G$ is not known. It is also called the  \textit{multi-index model} or \textit{dimension-reduction model}, since usually $K$ is much smaller than $p$ and we wish to treat $\xx_i^T \bbeta_1,\ldots,\xx_i^T \bbeta_K$ as new coordinates. A primary interest of this model is $\spann\{ \bbeta_1,\ldots,\bbeta_K\}$, denoted again by $\cS$, since it provides a pathway to dimension reduction, data visualization, estimation of $g$, and so on.

Let the regression mean function be
\begin{equation*}
\E(y_i | \xx_i) = g(\xx_i^T \bbeta_1,\ldots,\xx_i^T \bbeta_K).
\end{equation*}
Assume the function $g$ is twice differentiable, and $\xx_i \sim N(0,\bI_p)$. Well-known estimation methods include sliced  inverse regression \citep{Li91}, principal Hessian directions \citep{Li92}, etc. For example, under our assumptions, the principal Hessian directions (pHd) method seeks to estimate $\E [ y_i (\xx_i \xx_i^T - \bI_p) ]$ from the data. (Equivalently, it estimates $\E [ r_i (\xx_i \xx_i^T - \bI_p) ]$ where $r_i = y_i -  \beta^{LS}_0 - \xx_i^T \bbeta^{LS}$ is the residual after taking out a least square fit). Thus a natural construction is the same as \eqref{eq:vell}. Similar to the previous subsection, we may also consider utilizing first moments, or applying a nonlinear function $h$, which leads to the same form as in \eqref{eq:vell2}.  These constructions guarantee $\E \bv_\ell \in \cS$, due to Stein's identity.

As pointed out by \cite{Li92, Cook98}, a drawback of the pHd is the possibility of vanishing $\E [ y_i (\xx_i \xx_i^T - \bI_p) ]$ (or its rank is smaller than $K$). This may occur, for instance, if $\E [\nabla^2 g(\xx_i^T \bbeta_1,\ldots,\xx_i^T \bbeta_K)] = \bzero$, which unfortunately includes linear regression. Moreover, any linear trend in $g$ is missed by pHd, which is an unpleasant feature of pHd. For instance, the direction of $\bbeta_2$ is not captured by pHd in the following example:
\begin{equation*}
y_i = g_0(\xx_i^T \bbeta_1) + \xx_i^T \bbeta_2 + \epsilon_i.
\end{equation*}

A remedy for pHd is making transformation of $y_i$ before applying pHd, and some success is reported by \citep{Li92, SunIoaMon14}. In this regard, the vectors constructed through transformation, as suggested in \eqref{eq:vell2}, agrees with the aforementioned papers. Besides, our approach can combine transformed moments as before.

\subsection{Distributed estimation for heterogeneous datasets}\label{sec:distest}

As a last example, we consider an estimation problem in a modern setting. Suppose we have $m$ datasets stored on separate clusters or held by different laboratories/hospitals. Due to communication cost or privacy concerns, we wish to compute statistics locally for each datasets and aggregate these statistics at a central server without accessing the details of these distributed datasets.

Consider the problem of estimating a subspace $\cS$ as before.
Let $n$ be the total sample size; and for ease of exposition, we introduce i.i.d.\ multinomial variable $z_i \in \{1,2,\ldots,m\}$ ($i \in [n]$) that indicates which dataset the data unit indexed by $i$ belongs to. Let $\rho_\ell := \P(z_i = \ell) \in (0,1)$ be fixed. Then, the sample size of each dataset is roughly $\rho_\ell n$.

Suppose for each $\ell \in [m]$, the $\ell$th dataset consists of measurements of the form $\bff^{(\ell,k)}(\xx_i)$ for all $i$ such that $z_i = \ell$. Here, $\xx_i$ is a random quantity (may or may not observed) associated with, say, a subject with index $i$ in the $\ell$th laboratory with $k$th measurement. The total number of measurements for each subject in the $\ell$th laboratory is $K_\ell$. Each laboratory computes an average of these measurements: for any $\ell \in [m]$ and $k\in[K_\ell]$,
\begin{equation*}
\bv^{(\ell, k)} := \frac{1}{ | \cI_\ell | }\sum_{i \in \cI_\ell} \bff^{(\ell,k)}(\xx_i), \quad \text{where } \cI_\ell = \{ i: z_i=\ell \}
\end{equation*}
and contributes these vectors to the central server.

These functions $\bff^{(\ell,k)}$ may be different, which reflects different methods of measurements across laboratories. If we assume $\E \bv^{(\ell,k)} \in \cS$ and that $(\xx_i; z_i)$ are i.i.d., then we can use our framework to aggregate $\bv^{(\ell,k)}$ across $\ell$ and $k$. To do so, we define $\bff_{\ell,k} (\xx_i, z_i) = \bff^{(\ell,k)}(\xx_i) \bone\{ z_i = \ell \}$ and rewrite $\bv^{(\ell,k)}$ as
\begin{equation*}
\bv^{(\ell,k)} = \frac{n}{| \cI_\ell |} \cdot \frac{1}{n} \sum_{i=1}^n \bff_{\ell,k}(\xx_i, z_i).
\end{equation*}
Set $\bv_{\ell,k} := n^{-1} | \cI_\ell | \bv^{(\ell,k)}$. It is clear that $\E \bv_{\ell,k} \in \cS$, and our general framework \eqref{def:v1}--\eqref{def:v2} applies here.

The total number of vectors we have is $M:=K_1+\ldots+K_m$, and correspondingly $\bSigma^*$ and $\bW^*$ are of size $M \times M$. Note that due to the specific form of $\bff_{\ell,k} (\xx_i, z_i)$, $\bSigma^*$ must be a block diagonal matrix, so we only need to estimate each block $\bSigma_{\ell \ell}^* \in \R^{K_{\ell} \times K_{\ell}}$ according to \eqref{def:sigmahat}, which can be computed locally on each dataset once an initial $\hat \bU_0$ is given. To yield the final estimator, we set $\bV_\ell = [\bv_{\ell,1},\ldots,\bv_{\ell,K_{\ell}}] \in \R^{p \times K_\ell}$ and calculate
\begin{equation}\label{exp:distest}
\hat \bU = \mathrm{eigen}_r \left( \sum_{\ell=1}^m \bV_{\ell}  ( \hat \bSigma_{\ell \ell})^{-1} \bV_{\ell}^T \right)
\end{equation}
where $\mathrm{eigen}_r(\cdot)$ computes the top $r$ eigenvectors. Here $(\hat \bSigma_{\ell \ell})^{-1}$ can be viewed as a local weighting matrix that takes into account the covariance-like information of the $\ell$th dataset. A nice property of \eqref{exp:distest} is that its computation does not require data collection across datasets, and thus may be useful in privacy-sensitive situations.

A special case is $K_1=\ldots=K_m=1$, i.e., all laboratories conduct a single measurement for each of its subjects. In this simple scenario, the weight in \eqref{exp:distest} simply becomes a scalar.

We remark that in a recent work on distributed estimation for spiked covariance model \citep{Fan17}, a similar method as \eqref{exp:distest} is proposed, except there is no weight $\hat \Sigma_{\ell \ell}^{-1}$ before $\bv_\ell \bv_\ell^T$. While the regimes are different, aggregating
$\bv_\ell \bv_\ell^T$ seems to be the gist of both methods. Our method, moreover, suggests weighting by $\hat \Sigma_{\ell \ell}^{-1}$, which utilizes the variance-like information and thus may be preferable for aggregating heterogeneous datasets.

\section{Large sample properties}\label{sec:theory}
We will establish results about our subspace GMM estimator with an analysis under the classical  `fixed $p$ and $m$, large $n$' regime. In this section, to avoid confusion, we explicitly display the dependence on $n$, i.e., $\hat \bU_n = \hat \bU$, $\hat \bSigma_n = \hat \bSigma$, $\bW_n = \bW$, etc. Proofs can be found in the supplementary materials.

\subsection{Consistency and asymptotic normality}\label{sec:normal}

As the first part of our analysis, we will establish consistency and asymptotic normality, as is often done in the GMM literature. The optimality of asymptotic variance requires a block matrix assumption that could be restrictive in practice. However, if our optimality is gauged not in the original parameter space $O(p,r) / O(r)$, but in terms of the canonical angles between two subspaces, then our procedure is superior under fairly mild conditions (Section~\ref{sec:optproc}).

We embark on our analysis by making a few assumptions. First, we assume that the vectors $\bv_1, \ldots, \bv_m$ (or equivalently $\bff_1(\xx_i,y_i), \ldots, \bff_m(\xx_i,y_i)$) contain sufficient information of $\cS$, in the sense that these vectors, in expectation, span the subspace $\cS$. Note that each individual $\E \bv_\ell$ lies in $\cS$, so an equivalent assumption is stated below in terms of the dimension. In general, this is a mild assumption, since $m \ge \dim(\cS) = r$.

\begin{ass}[nondegeneracy]\label{ass:span}
Suppose $\dim(\spann\{ \E \bv_1, \ldots, \E \bv_m \}) = r$.
\end{ass}

This assumption is a prerequisite for any reasonable estimator of $\cS$. Indeed, if $\dim(\cS) > \dim(\spann\{ \E \bv_1, \ldots, \E \bv_m \})$, even with infinite sample size, it is impossible to uniquely determine $\cS$. In this regard, Assumption~\ref{ass:span} is an identifiability assumption for our problem formulated in \eqref{def:v1}--\eqref{def:v2}.

To state a general consistency result, we consider the following mild condition on the weighting matrix. In particular, in our subspace GMM estimation procedure, as long as $\bSigma^*$ is invertible, this assumption is satisfied for both the initial estimator (where $\bW_n$ is a constant matrix) and the final GMM estimator (where $\bW_n \xrightarrow{p} (\bSigma^*)^{-1}$).

\begin{ass}[limiting weight matrix]\label{ass:W}
Suppose $\bW_n$ converges in probability to a non-random positive definite matrix $\bW^* \succ \bzero$ in $\R^{m \times m}$.
\end{ass}

Similar to $\bar \bW$ in \eqref{def:Wbar}, we define a limiting block matrix $\bar \bW^* := \bW^* \otimes \bI_p \in \R^{\bar m \times \bar m}$. It follows from Lemma~\ref{lem:Kron} that $\bar \bW^*$ is positive definite under Assumption~\ref{ass:W}. Since the vectors $\E \bv_1, \ldots, \E \bv_m$ span the full subspace $\cS$ under Assumption~\ref{ass:span}, we can show that $\bU^*$ is the only matrix in $O(p,r)$, up to rotation, such that $Q^*(\bU) := [ \E \bg(\bU)]^T \bar \bW^* [\E\bg(\bU)]$ is minimized (the minimum is $0$). Thus, it is natural to expect $\hat \bU_n$, as a minimizer of $Q(\bU)$, to be close to $\bU^*$ up to rotation when $n$ is large.

The first theorem is a reassuring consistency result. We consider a generic weighting matrix $\bW^*$ without specifying particular choices.

\begin{thm}\label{thm:consist1}
Under Assumptions~\ref{ass:span} and \ref{ass:W}, there exists a sequence of orthogonal matrices  $\bR_1, \bR_2, \bR_3, \ldots \in O(r)$ such that
\begin{equation*}
\hat \bU_n \bR_n \xrightarrow{p} \bU^*, \qquad \text{as} ~ n \to \infty,
\end{equation*}
where $\xrightarrow{p}$ means convergence in probability.
\end{thm}

We remark that if we consider the regime where $m$ also grows, typically we have consistency if $m = o(n)$ with additional regularity assumptions on $\bW^*$ \citep{KoeMac99, Donald03}.

Next, in accordance with our procedure where block weighting matrix is used, we consider the following assumption of block matrix forms for the covariance matrix $\bar \bS^*$.

\begin{ass}[block-wise covariance]\label{ass:uniform}
Suppose the covariance matrix of the concatenated vector
\begin{equation*}
\bar \bff(\xx_i, y_i) := [\bff_1(\xx_i, y_i); \ldots ; \bff_m(\xx_i, y_i)] \in \R^{\bar m},
\end{equation*}
denoted by $\bar \bS^*$, has the following block matrix form: $\bar \bS^* = \bS^* \otimes \bI_p$, where $\bS^* \in \R^{m \times m}$ is positive definite.
\end{ass}

Under this assumption, there is a simple connection between $\bS^*$ and $\bSigma^*$.
\begin{lem}\label{lem:sigmaS}
Under Assumption~\ref{ass:uniform}, we have $\bSigma^* = (p-r) \bS^*$.
\end{lem}

Similar to the classical theory in GMM, we will establish asymptotic normality of our estimator $\hat \bU_n \bR_n$ (corrected by a rotation $\bR_n$). The asymptotic covariance of our GMM estimator is a function of the weighting matrix $\bW_n$. Once an explicit expression is obtained, it is then easy to justify, in terms of asymptotic efficiency, the optimality of our estimation procedure.

An important departure from the classical GMM theory is that, our estimator must satisfy the constraint $\hat \bU_n \in O(p,r)$, which is a manifold of dimension $pr - r(r+1)/2$, and that the estimator is determined up to rotation. For this reason, instead of the difference $\hat \bU_n \bR_n - \bU^*$, we measure our estimator through $\bP_{\hat \bU_n}^\bot(\hat \bU_n \bR_n - \bU^*)$, where $\bP_{\hat \bU_n}^\bot:= \bI_p - \hat \bU_n \hat \bU_n^T \in \R^{p \times p}$. A desirable property of applying a projection $\bP_{\hat \bU_n}^\bot$ is that it does not depend on the choice of $\bR_n$, as $\bP_{\hat \bU_n}^\bot\hat \bU_n = \bzero$ always holds, so we can also write $\bP_{\hat \bU_n}^\bot (\hat \bU_n \bR_n - \bU^*)  =  -\bP_{\hat \bU_n}^\bot \bU^* $. The following result also serves as an intermediate result towards optimality under $d(\hat \cS, \cS)$. Below, the asymptotic variance is a degenerate matrix, but nevertheless, it is the smallest one in terms of the (generalized) inequality $\succeq$.

\begin{thm}\label{thm:normality}
Define $\bP_{\hat \bU_n}^\bot:= \bI_p - \hat \bU_n \hat \bU_n^T \in \R^{p \times p}$, $\bP_{\bU^*}^\bot:= \bI_p - \bU^* (\bU^*)^T \in \R^{p \times p}$, and $\bG^* := [ \E \bv_1, \ldots, \E \bv_m]^T \bU^* \in \R^{m \times r}$. Under Assumptions~\ref{ass:span} and \ref{ass:W}, we have
\begin{align}
&\sqrt{n} \, \vect( \bP_{\hat \bU_n}^\bot \bU^*) \xrightarrow{d} N(\bzero, (\bA^* \otimes \bP_{\bU^*}^\bot)\, \bar \bS^* \, ((\bA^*)^T \otimes \bP_{\bU^*}^\bot)), \label{asym:1} \\
&\text{where} ~~~~ \bA^* := [(\bG^*)^T \bW^* \bG^* ]^{-1} (\bG^*)^T \bW^*, \notag
\end{align}
where $\vect(\cdot)$ stacks matrix columns into a vector. If we suppose, in addition, Assumption~\ref{ass:uniform}, then the choice $\bW^* = (\bS^*)^{-1}$ or $\bW^* = (\bSigma^*)^{-1}$ leads to the smallest asymptotic variance, in the following sense
\begin{equation}\label{ineq:long}
(\bA^* \otimes \bP_{\bU^*}^\bot)\, \bar \bS^* \, ((\bA^*)^T \otimes \bP_{\bU^*}^\bot) \succeq  [(\bG^*)^T (\bS^*)^{-1} \bG^* ]^{-1} \otimes \bP_{\bU^*}^\bot, \qquad \forall \, \bW^* \succ \bzero.
\end{equation}
\end{thm}

Note that the asymptotic variance is invariant to the scaling of $\bW^*$; in particular, the choice $\bW^* = (\bS^*)^{-1}$ or $\bW^* = (\bSigma^*)^{-1}$ leads to the same asymptotic variance, in light of Lemma~\ref{lem:sigmaS}. Note also that the right-hand side of \eqref{ineq:long} is a symmetric matrix, due to basic properties of Kronecker products (see Lemma~\ref{lem:Kron} (iii)). The covariance-like matrix $\bSigma^*$ is essentially a compact form of $\bar \bS^*$, encoding the covariances through summation over $p$ coordinates. A natural question is that, without the block-wise assumption, whether the choice $\bW^* = (\bSigma^*)^{-1}$ is optimal for a simpler criterion. The next subsection gives an affirmative answer.

\subsection{Optimality of estimation procedure}\label{sec:optproc}

In this subsection, we measure the difference between the estimated subspace $\hat \cS = \spann(\hat \bU_n)$ and the true $\cS = \spann(\bU^*)$ using the notion \textit{canonical angles}, which is a generalization of angles between two vectors.

Let $\sigma_1, \ldots, \sigma_r$ be the singular values of $\hat \bU_n^T \bU^*$. Then, the canonical angles $\theta_1,\ldots, \theta_r \in [0, \pi/2]$ between $\hat \bU_n$ and $\bU^*$ are $\theta_k = \arccos(\sigma_k)$, $k \in [r]$. Also denote $\sin \bTheta_n(\bW_n) = \diag\{ \sin \theta_1,\ldots, \sin \theta_r\} \in \R^{r \times r}$, where we show explicitly the dependency on $\bW_n$ and $\Theta_n$r. Note that rotational invariance of singular values ensures that the canonical angles do not depend on the basis chosen to represent subspaces. The canonical angles are closely related to a common distance between $\spann(\hat \bU_n)$ and $\spann(\bU^*)$ \citep{Ste90, Ste98}:
\begin{equation}\label{eq:mse}
d(\hat \cS, \cS)^2 = \| \hat \bU_n \hat \bU_n^T - \bU^* (\bU^*)^T  \|_F^2 = 2 \| \sin \bTheta_n(\bW_n) \|_F^2.
\end{equation}

As briefly discussed in the previous subsection, we may remove the restrictive block-form Assumption~\ref{ass:uniform} if the gauge is a different quantity. To this end, given a weighting matrix $\bW_n \succ \bzero$, we gauge $\hat \bU = \hat \bU(\bW_n)$ through an $r \times r$ positive semidefinite matrix
\begin{equation}\label{def:psi}
\bPsi_n(\bW_n)  = (\bU^*)^T \bP_{\hat \bU_n}^\bot \bU^* = \bI_r - (\bU^*)^T \hat \bU_n  (\hat \bU_n)^T \bU^*.
\end{equation}
The eigenvalues of $\bPsi_n(\bW_n)$ are determined by the canonical angles between $\hat \bU_n$ and $\bU^*$. In particular, a useful identity is given below. See the derivation in the proof of Theorem~\ref{thm:opt2}.
\begin{equation}\label{eq:connect}
\Tr \big( \bPsi_n(\bW_n) \big) =  \| \sin \bTheta_n(\bW_n) \|_F^2 = \frac{1}{2}\, d(\hat \cS, \cS)^2.
\end{equation}

To see how we may relax Assumption~\ref{ass:uniform}, we note that $\bSigma^*$, as defined in \eqref{def:sigma0}, captures covariance-like information of the vectors $\bff_1(\xx_i,y_i),\ldots, \bff_m(\xx_i,y_i)$ in an average sense. More precisely, we have the following identity:
\begin{equation}\label{def:sigma}
\Sigma_{j\ell}^* = \Tr\big[ \cov (\bP_{\bU^*}^\bot\bff_j(\xx_i,y_i),  \,\bP_{\bU^*}^\bot \bff_{\ell}(\xx_i,y_i) )   \big], \quad \forall\, j,\ell \in [m].
\end{equation}
Here, recall that $\bP_{\bU^*}^\bot = \bI_r - \bU^* (\bU^*)^T$ is a deterministic matrix, and that $ \bP_{\bU^*}^\bot \, \bff_j(\xx_i,y_i)$ has zero mean. Since $\bPsi_n(\bW_n)$ is determined by summation over $p$ coordinates, it is reasonable to expect that we only require a condition on $\bSigma^*$, rather than the big $\bar m \times \bar m$ matrix $\bar \bS$.

This intuition leads to Assumption~\ref{ass:sigma} below, which greatly relaxes Assumption~\ref{ass:uniform}. It is possible to further relax this assumption (see Section~\ref{sec:singular}), but for now we will content ourselves with it.
\begin{ass}[invertibility]\label{ass:sigma}
The matrix $\bSigma^*$ defined in \eqref{def:sigma} is not singular.
\end{ass}

We will use a notion called \textit{asymptotic expectation} to assess and compare the quality of $\bPsi_n(\bW_n)$ for different $\bW_n$. Formally, similar to \cite{Shao03}, we define the asymptotic expectation below.

\begin{defn}\label{def:aE}
Let $\{\bxi_n\}$ be a sequence of random vectors and $\{a_n\}$ be a sequence of positive numbers satisfying $a_n \to \infty$ or $a_n \to a >0$. If $a_n \bxi_n \xrightarrow{d} \bxi$ and $\E \| \bxi \|_2 < \infty$, then $\E \bxi / a_n$ is called the asymptotic expectation of $\bxi_n$, which is denoted by $\aE(\bxi_n)$.
\end{defn}
Note that if $\E \bxi \neq \bzero$, then the asymptotic expectation is unique up to a $(1 + o(1))$ factor. This is a consequence of Prop.\ 2.3 of \cite{Shao03}. Roughly speaking, the asymptotic expectation is a `weaker' version of the usual expectation, because weak convergence  bypasses the issue of potential extreme values of $\bxi_n$, which are often caused by its singularities. If certain sort of uniform convergence is guaranteed, we will recover the usual expectation, as stated in the next theorem.

\begin{thm}\label{thm:opt2}
Under Assumptions~\ref{ass:span}, \ref{ass:W} and \ref{ass:sigma}, the asymptotic expectation\footnote{Statements and inequalities involving asumptotic expectation are up to a $1+o(1)$ factor, due to the nature of its definition.} of $\bPsi_n(\bW^*)$ is minimized at $\bW^* = (\bSigma^*)^{-1}$, that is, for any limiting $ \bW^* \succ \bzero$,
\begin{equation*}
\aE\Big(\bPsi_n((\bSigma^*)^{-1}) \Big) = \frac{1}{n} \,[(\bG^*)^T (\bSigma^*)^{-1}\bG^* ]^{-1}  \preceq \aE \Big(\bPsi_n(\bW^*) \Big) = \aE \Big(\bPsi_n(\bW_n) \Big) .\end{equation*}
If, in addition, $\{ n\| \bPsi_n(\bW^*) \|_F \}_{n} $ and $\{ n\| \bPsi_n((\bSigma^*)^{-1}) \|_F \}_n $ are both uniformly integrable, then
\begin{equation*}
 \E\,\bPsi_n((\bSigma^*)^{-1}) \preceq (1+o(1)) \cdot \E\, \bPsi_n(\bW^*).
\end{equation*}
Moreover, the same holds if we replace $\bPsi_n((\bSigma^*)^{-1})$ by $\bPsi_n((\hat \bSigma_n)^{-1})$.
\end{thm}

In the above theorem, the asymptotic expectations of $\bPsi_n((\bSigma^*)^{-1})$ and $\bPsi_n((\hat \bSigma_n)^{-1})$ are essentially the same (up to a $1+o(1)$ factor). This is because the asymptotic expectations are determined by the limiting weighting matrix, and $(\hat \bSigma_n)^{-1}$ implemented in our procedure converges in probability to $(\bSigma^*)^{-1}$.

Using \eqref{eq:connect}, we obtain an optimality result for the final estimator $\hat \bU_n^{\GMM} = \hat \bU_n(\hat \bSigma_n^{-1})$ produced by our estimation procedure in Section~\ref{sec:procedure}.

\begin{cor}\label{cor:procedure}
Under Assumptions~\ref{ass:span}, \ref{ass:W} and \ref{ass:sigma}, the matrix $\hat \bSigma_n$ computed in Step 2 of our procedure is invertible with probability $1-o(1)$, and for any limiting $\bW^* \succ \bzero$,
\begin{align*}
\aE\left( \left\| \hat \bU_n^{\GMM} [\hat \bU_n^{\GMM}]^T - \bU^* (\bU^*)^T  \right\|_F^2  \right) \le \aE\left( \left\| \hat \bU_n(\bW_n) [\hat \bU_n(\bW_n)]^T - \bU^* (\bU^*)^T  \right\|_F^2  \right).
\end{align*}
\end{cor}

Similar to Theorem~\ref{thm:opt2}, the asymptotic expectation can be replaced by the usual expectation if uniform integrability is satisfied.

Our results in this subsection rely on Assumption~\ref{ass:sigma}, namely, the invertibility of $\bSigma^*$. However, if $m$ is not very small, $\bSigma^*$ may be singular or nearly singular in practice. Before addressing this issue and justifying the de-noising step suggested in \eqref{def:generalinv}, we first focus on a more practical aspect---estimating the dimension of the subspace $r$ if it is not known in advance---in the next subsection.

\subsection{Estimating subspace dimension $r$}\label{sec:dimest}

In Section~\ref{sec:subspGMM} and \ref{sec:normal}, we assumed that the subspace dimension $r$ is known. However, in practice, this is usually unknown. For example, in finite mixture models, $r$ is the number of mixtures, which typically needs to be estimated from the data. Therefore, in these cases, it is of interest to determine $r$ prior to obtaining the subspace GMM estimator.

Fortunately, the simple form of \eqref{opt:eigen} allows us to estimate $r$ consistently using only the eigenvalues of $\bV_n \bW_n \bV_n^T$. For a symmetric matrix $\bA$, let us denote by $\lambda_j(\bA)$ the $j$th largest eigenvalue of $\bA$.

\begin{thm}\label{thm:dimest}
Suppose $r<p$, and Assumptions~\ref{ass:span} and ~\ref{ass:sigma} hold. If $\bW_n \xrightarrow{p} \bW^* = (\bSigma^*)^{-1}$, where $\bSigma^*$ is defined in \eqref{def:sigma0}, then $\sum_{j=r+1}^p \lambda_j(  \bV_n \bW_n \bV_n^T ) = O_P(n^{-1})$. Moreover, assuming Assumption~\ref{ass:uniform} in addition, we have
\begin{equation*}
n(p-r) \sum_{j=r+1}^p \lambda_j(  \bV_n \bW_n \bV_n^T ) \xrightarrow{d} \chi^2_{(p-r)(m-r)}.
\end{equation*}
In particular, the choice $\bW_n = (\hat \bSigma_n)^{-1}$ in our procedure satisfies $\bW_n \xrightarrow{p} (\bSigma^*)^{-1}$.
\end{thm}
This result predicts a large eigen-gap between top $r$ eigenvalues and the remaining eigenvalues of $\bV_n \bW_n \bV_n^T$ in the large sample regime: as $n$ grows to $\infty$, the sum of the remaining eigenvalues scale with $1/n$, whereas, under Assumption~\ref{ass:span}, $\lambda_r(  \bV_n \bW_n \bV_n^T )$ remains a constant order. This suggests a simple threshold-based method:
\begin{equation*}
\hat r_{\tau} = \argmax_{k} \{ k \, | \,\lambda_k(  \bV_n \bW_n \bV_n^T ) > \tau_n \},
\end{equation*}
where $\tau_n$ is the threshold value with $\tau_n = o(1)$ and $n\tau_n \to \infty$. In cases where Assumption~\ref{ass:uniform} holds and the chi-squared distribution is a suitable approximation, one can also use the following estimator, which is similar to \cite{Li91} in spirit.
\begin{equation}\label{def:rhateta}
\hat r_{\eta} = \argmin_{k} \{ k \,| \, n(p-k)\, \sum_{j > k} \lambda_j(  \bV_n \bW_n \bV_n^T ) \le \eta_n(k) \}.
\end{equation}
The parameter $\eta_n(k)$ is related to the quantiles of $\chi^2_{(p-k)(m-k)}$. In practice, for example, one may choose $\eta_n(k)$ to be the $95$th quantile of $\chi^2_{(p-k)(m-k)}$.  The following corollary establishes the consistency of both estimators.
\begin{cor}\label{cor:dimest}
Suppose $r<p$, Assumptions~\ref{ass:span} and ~\ref{ass:sigma} hold,  and $\bW_n \xrightarrow{p} (\bSigma^*)^{-1}$. Then, with $\tau_n = o(1)$ and $n\tau_n \to \infty$, the estimator $\hat r_\tau$ is consistent; with $\eta_n(r) = o(n)$ and $\eta_n(r) \to \infty$, the estimator $\hat r_\eta$ is consistent.
\end{cor}

In the literature, especially in factor models, eigenvalues are the basis of many methods for rank or dimension estimation (e.g., the number of factors in factor models). An early work is the scree test method \citep{Cat66}, which sorts and plots the eigenvalues in descending order. Many recent works on dimension estimation focus on factor models, including \cite{BaiNg02}, in which one minimizes an objective function that is similar to information criterion, and \cite{Ona10,LamYao12,AhnHor13}, in which one use differences or ratios of eigenvalues to determine the number of factors, etc.

Though being related, these recent works derive sophisticated methods based on different models and typically consider different regimes ($p$ grow with $n$).

\subsection{Redundancy and singular $\bSigma^*$}\label{sec:singular}

In this subsection, our goal is to further relax the non-singular Assumption~\ref{ass:sigma}, and justifies the use of thresholding in Step 4 of our estimation procedure. Often, we derive the vectors $\bv_\ell$ from different approaches, hoping to improve the estimation of $\cS$. However, in the process, we may inadvertently introduce redundant vectors that lead to a singular $\bSigma^*$. In our simulations, for example, we do observe such singularity issue. Nevertheless, if the singularity of $\bSigma^*$ is solely caused by redundant vectors, then, as we will show, we can simply resort to the thresholding without losing any optimality guarantee.

Let us first define precisely what we mean by `redundancy'. For simplicity, let us write $\bff_\ell := \bff_\ell(\xx_i, y_i)$, suppressing the dependency on $i$ due to i.i.d.\ data. Recall the notation $\bF = [\bff_1,\ldots,\bff_m] \in \R^{p \times m}$, where we suppress the subscript $i$ due to the i.i.d.\ assumption. Let us partition $\bF$ and $\bSigma^*$ as follows:
\begin{align}\label{eq:partition}
\begin{split}
&\bF = [ \bF_1, \bF_2 ], \quad \text{where} ~ \bF_1 \in \R^{p \times m_1}, \bF_2 \in \R^{p \times m_2}, \\
& \bSigma^* = \left[ \begin{array}{cc}
\bSigma^*_{11} & \bSigma^*_{12} \\
\bSigma^*_{21} & \bSigma^*_{22}
\end{array} \right], \quad
\text{where} ~ \bSigma^*_{ab} \in \R^{m_a \times m_b}, a,b \in \{1,2\}.
\end{split}
\end{align}
By construction, the dimensions must satisfy $m_1 + m_2 = m$. We require that such partition separates the non-redundant vectors (columns of $\bF_1$) from the redundant vectors (columns of $\bF_2$), as will soon be explained. We require $\bF_1$ to be well behaved in the sense that $\bSigma_{11}^*$ is invertible. Moreover, we let
\begin{equation}\label{eq:F2}
\tilde \bF_2 := \bF_2 - \bF_1 (\bSigma_{11}^*)^{-1} \bSigma_{12}^* \in \R^{p \times m_2}
\end{equation}
be the part of $\bF_2$ not explained by $\bF_1$. This construction ensures $\bF_1$ and $\tilde \bF_2$ have zero cross-covariance, i.e., $\E\, \bF_1^T \bF_2 = \bzero$. Also, as seen in Section~\ref{sec:combine}, the expectation of any linear combination of vectors $\bff_\ell$ (or $\bv_\ell$) always lies in $\cS$, and thus $\E \,\tilde \bF_2$ satisfies the condition \eqref{def:v2}. We call $\bF_2$ redundant if $\E \,\tilde \bF_2 = \bzero$; in other words, in expectation, $\bF_2$ is fully explained by $\bF_1$. This separation of well-behaved vectors and redundant vectors is summarized the following assumption.
\begin{ass}[partial invertibility]\label{ass:redundancy}
Suppose $0 \le m_2 < m$, and that $\bF$ and $\bSigma^*$ can be partitioned as in \eqref{eq:partition} such that $\bSigma_{11}^*$ is invertible, and that $\tilde \bF_2$, as defined in \eqref{eq:F2}, satisfies $\E\, \tilde \bF_2 = \bzero$.
\end{ass}
Note that we permit $m_2=0$, i.e., $\bSigma^*$ is non-singular and no partition is needed, which includes Assumption~\ref{ass:sigma} as a special (less general) case. This assumption relaxes the invertibility assumption in that it allows scenarios where some $\bff_\ell$ are fully explained by others, often due to overlapping information practitioners introduce when collecting these $\bff_\ell$. In particular, this assumption is satisfied if $\bSigma_{11}^*$ is invertible and $\bF_2$ is a linear transformation of $\bF_1$.

We also note that our notion of redundancy is related to that in \cite{Bre99}: in both cases, the redundant part is explained by the other part in a way similar to \eqref{eq:F2}, but our covariance-like matrix $\bSigma^*$ is tailored to the special structure here.


We will generalize Section~\ref{sec:dimest} by showing that the optimal limiting weighting matrix is $(\bSigma^*)^{+}$, i.e., the pseudoinverse of $\bSigma^*$, and that the thresholding step \eqref{def:generalinv} guarantees $\bW_n \xrightarrow{p} (\bSigma^*)^{+}$ for a suitable parameter $\delta_n$. The optimality of $(\bSigma^*)^{+}$ is established in a slightly larger family of limiting weighting matrices:
\begin{equation*}
\cW := \{ \bW^* \in \R^{m \times m}  | \, \bW^* \succeq \bzero \,, (\bG^*)^T \bW^* \bG^* \succ \bzero \}.
\end{equation*}
This family $\cW$ includes all positive semidefinite matrices. We can show that $(\bG^*)^T (\bSigma^*)^{+}\bG^*$ is always positive definite, and thus $(\bSigma^*)^{+} \in \cW$. Moreover,  any limiting weighting matrix in $\cW$ leads to a consistent estimator for $\cS$, which is a generalization of Theorem~\ref{thm:consist1} (see Theorem~\ref{thm:consist2} in the supplementary materials). The next result is a formal optimality statement.

\begin{thm}\label{thm:opt3}
Suppose Assumptions~\ref{ass:span} and \ref{ass:redundancy} hold, and $\bW_n \xrightarrow{p} (\bSigma^*)^{+}$. Then, $(\bG^*)^T (\bSigma^*)^{+}\bG^* \succ\bzero$, and for any sequence $\{ \bW_n'\}$ with $ \bW_n' \xrightarrow{p} \bW^* \in \cW$,
\begin{equation}\label{ineq:opt3}
\aE\Big(\bPsi_n(\bW_n) \Big) = \frac{1}{n} \,[(\bG^*)^T (\bSigma^*)^{+}\bG^* ]^{-1}  \preceq \aE \Big(\bPsi_n(\bW^*) \Big) = \aE \Big(\bPsi_n(\bW_n') \Big).
\end{equation}
In particular, $\bW_n$ in \eqref{def:generalinv} with $\delta_n = o(1)$ and $\sqrt{n}\, \delta_n \to \infty$ satisfies $\bW_n \xrightarrow{p} (\bSigma^*)^{+}$.
\end{thm}
This theorem is a generalization of Theorem~\ref{thm:opt2}. Indeed, if $\bSigma^*$ is invertible, then we recover Theorem~\ref{thm:opt2}. Similarly, if we set $\hat \bU_n^{\GMM} = \hat \bU_n(\bW_n)$ where $\bW_n$ is computed from \eqref{def:generalinv}, we obtain a generalized corollary.
\begin{cor}\label{cor:opt3}
Suppose Assumptions~\ref{ass:span} and \ref{ass:redundancy} hold, and we compute $\bW_n$ according to \eqref{def:generalinv} with $\delta_n = o(1)$ and $\sqrt{n}\, \delta_n \to \infty$. Then, for any sequence $\{ \bW_n'\}$ with $ \bW_n' \xrightarrow{p} \bW^* \in \cW$,
\begin{align}\label{ineq:cor3}
\aE\left( \left\| \hat \bU_n^{\GMM} [\hat \bU_n^{\GMM}]^T - \bU^* (\bU^*)^T  \right\|_F^2  \right) \le \aE\left( \left\| \hat \bU_n(\bW_n') [\hat \bU_n(\bW_n')]^T - \bU^* (\bU^*)^T  \right\|_F^2  \right).
\end{align}
\end{cor}
We remark that, in practice, a larger $\delta_n$ may be used to select more important $\bff_\ell$ (or their linear combinations), apart from eliminating redundant $\bff_\ell$. This may be useful when $m$ is large. See more discussion in Section~\ref{sec:disc}.

Note also that asymptotically, adding more moments does not increase errors: if $\bV' \in \R^{p \times m'}$ is a submatrix of $\bV$ with $m' < m$, then under similar assumptions, the asymptotic error of the GMM estimator is non-increasing in $m$. That is, denoting the left-hand side of \eqref{ineq:cor3} by $\aE(\mathrm{err}^m)$, we have $\aE(\mathrm{err}^{m}) \le \aE(\mathrm{err}^{m'})$. This is because we can constrain $\bW^*$ to have nonzero values only in a $m' \times m'$ block.

\subsection{Augmenting regular moments}\label{sec:augment}
It is clear that, since the columns of $\bV_n$ are the moments of the data (with or without transformations), the matrix $\bV_n \bW_n \bV_n^T$ takes into consideration all pairwise products of these moments. However, it is not explicitly clear whether the moments themselves are included. The purpose of this subsection is to show that, implicitly, it also covers the moments themselves. 

To elaborate on this point, let us consider the simple factor model we discussed in Section~\ref{sec:factormodel}. Suppose we choose $\bV_n = [\bv_1, \bv_2,\ldots,\bv_{p+1}]$ where $\bv_1$ and $\bv_{1+\ell}$ ($\ell \in [p]$) are vectors given in \eqref{eq:vfactor}. They consist of the first moments and the second moments of $\xx_i$. Equivalently, we can write $\bV_n = [\bar \xx, n^{-1}\bX^T \bX - \sigma^2 \bI_p]$, where $\bX = [\xx_1,\ldots, \xx_n]^T \in \R^{n \times p}$. The GMM estimator finds the top eigenspace of $\bV_n \bW_n \bV_n^T$.

The above estimate is predominately based on the pairwise products of the second moment information.  A natural question is if the first moment condition, based on $n^{-1} \bX^T \bX$, provides additional information. See Section~\ref{sec:extension} also for discussions. A natural aggregation of the information based on the first and second moment conditions is 
\begin{equation*}
\eigen_r\left( \kappa\, n^{-1} \bX^T \bX +  \bV_n \bW_n \bV_n^T \right)
\end{equation*}
with a suitably chosen parameter $\kappa \ge 0$. However, we shall prove that the most efficient estimator we can hope to obtain is simply the GMM estimator (or equivalently $\kappa = 0$).  This means that the information contained in the traditional PCA-based information $n^{-1} \bX^T \bX$ is contained in our second moment method $\eigen_r\left( \bV_n \bW_n \bV_n^T \right)$.

%


\begin{thm}\label{thm:augment}
Assume the same conditions as in Corollary~\ref{cor:opt3}. Also suppose that $\bM_n \in \R^{p \times p}$ is a symmetric submatrix of $\bV_n$ satisfying $\rank(\E \bM_n) = r$ and $\E \bM_n \succeq \bzero$. For any parameter $\kappa \ge 0$ and any sequence $\{ \bW_n'\}$ with $\bW_n' \xrightarrow{p} \bW^* \in \cW$, define
\begin{align*}
\tilde \bU_n := \tilde \bU_n(\kappa,\bW_n') &= \eigen_r\left(\kappa \bM_n +  \bV_n \bW_n' \bV_n^T\right),
\end{align*}
Then, our GMM estimator $\hat \bU_n^\GMM$ satisfies
\begin{equation*}
\aE\left( \left\| \hat \bU_n^{\GMM} [\hat \bU_n^{\GMM}]^T - \bU^* (\bU^*)^T  \right\|_F^2  \right) \le \aE\left( \left\| \tilde \bU_n \tilde \bU_n^T - \bU^* (\bU^*)^T  \right\|_F^2  \right). 
\end{equation*}
\end{thm}

Note that, the condition $\rank(\E \bM_n) = r$ essentially requires that $\bM_n$ contains enough information about $\cS$, which is similar to Assumption~\ref{ass:span}. Also note that, in general, $\kappa \ge 0$ is necessary for consistency. The proof idea is that the dominant part of $\bM_n$ (that is, its best rank-$r$ approximation) can be absorbed into $\bV_n \bW_n' \bV_n^T$, while its remainder is negligible. See supplementary materials for details.

This result connects our method to \cite{LetPel17}, in which eigenvectors of $n^{-1} \bX \bX^T + \gamma \bar \xx \, \bar \xx^T$ are studied. Our result is more general in the sense that we consider a weighting matrix instead of a single parameter $\gamma$.

\section{Numerical simulations}\label{sec:sim}
In this section, we show how our methods may extract information from a pool of commonly used methods, under two kinds of models in Section~\ref{sec:factormodel}--\ref{sec:mixturemodel}: factor models (Example 1) and mixture models (Example 2). Both examples assume a known $r$. In the supplementary materials, we give additional simulation results for the multiple index models (Section~\ref{sec:indexmodel}) and consistent dimension estimation (Section~\ref{sec:dimest}).

In all simulations, the hard-thresholding parameter $\delta_n$ is fixed at $0.01$.

\subsection*{Example 1: first moments may help}

Let us consider the simple factor model \eqref{model:factor}. We generate each entry of $\bB$ independently from $N(0,1)$, and keep it fixed as the loading matrix throughout the data generating process. We also generate $\bz_i \sim N(\bmu_\bz, \bI_r)$ and $\bepsilon_i \sim N(\bzero,\sigma^2 \bI_p)$ independently, and then compute $\xx_i$. The left singular vectors of $\bB$ are computed and set to be $\bU^*$.

We fix $p=10$, $r = 2$ and $\sigma = 2$, and set $\bmu_\bz = (\mu, -\mu)^T$, where $\mu$ is a parameter. We consider two experiments: (1) fix $n = 500$, and let $\mu$ run through $\{0,0.5,1,1.5,\ldots,4\}$; (2) fix $\mu = 2$, and let $n$ run through $\{100,200,400,800,1600,3200\}$. The first experiment studies how much information we can extract from the first moments with a varying $\mu$, and the second experiment studies the effect of sample sizes on the performance of methods under investigation.

We compare the performance of three estimators: (i) \textit{standard}---the first one is the top $r$ eigenvectors of the sample covariance of $\by_i$, which is the standard method; (ii) \textit{GMM full}---the second one is the two-step estimator proposed in Section~\ref{sec:procedure}; (iii) \textit{GMM diagonal}---the third one keeps only the diagonal entries of $\hat \bSigma$ and sets the rest to zero, and elsewhere follows the same procedure as \textit{GMM full}. Both GMM estimators (ii) and (iii) are constructed from the $p+1$ vectors in \eqref{eq:vfactor}.

The performance of estimators is assessed over $100$ independent simulations and measured through the estimation error $\E \| \hat \bU\hat \bU^T - \bU^* (\bU^*)^T \|_F$, where $\hat \bU$ is any of the three estimators. We compute this error by averaging over all $100$ simulations, and we also compute the standard errors of the averages. The results are presented in Figure~\ref{fig:ex1}. We have also considered other estimation errors, including $\E \| \hat \bU\hat \bU^T - \bU^* (\bU^*)^T \|_F^2$, $\E \| \hat \bU\hat \bU^T - \bU^* (\bU^*)^T \|_2$ and $\E \| \hat \bU\hat \bU^T - \bU^* (\bU^*)^T \|_2^2$. The results for these errors are similar, and figures are omitted here.

\begin{figure*}[t!]
    \centering
    \includegraphics[scale=0.5]{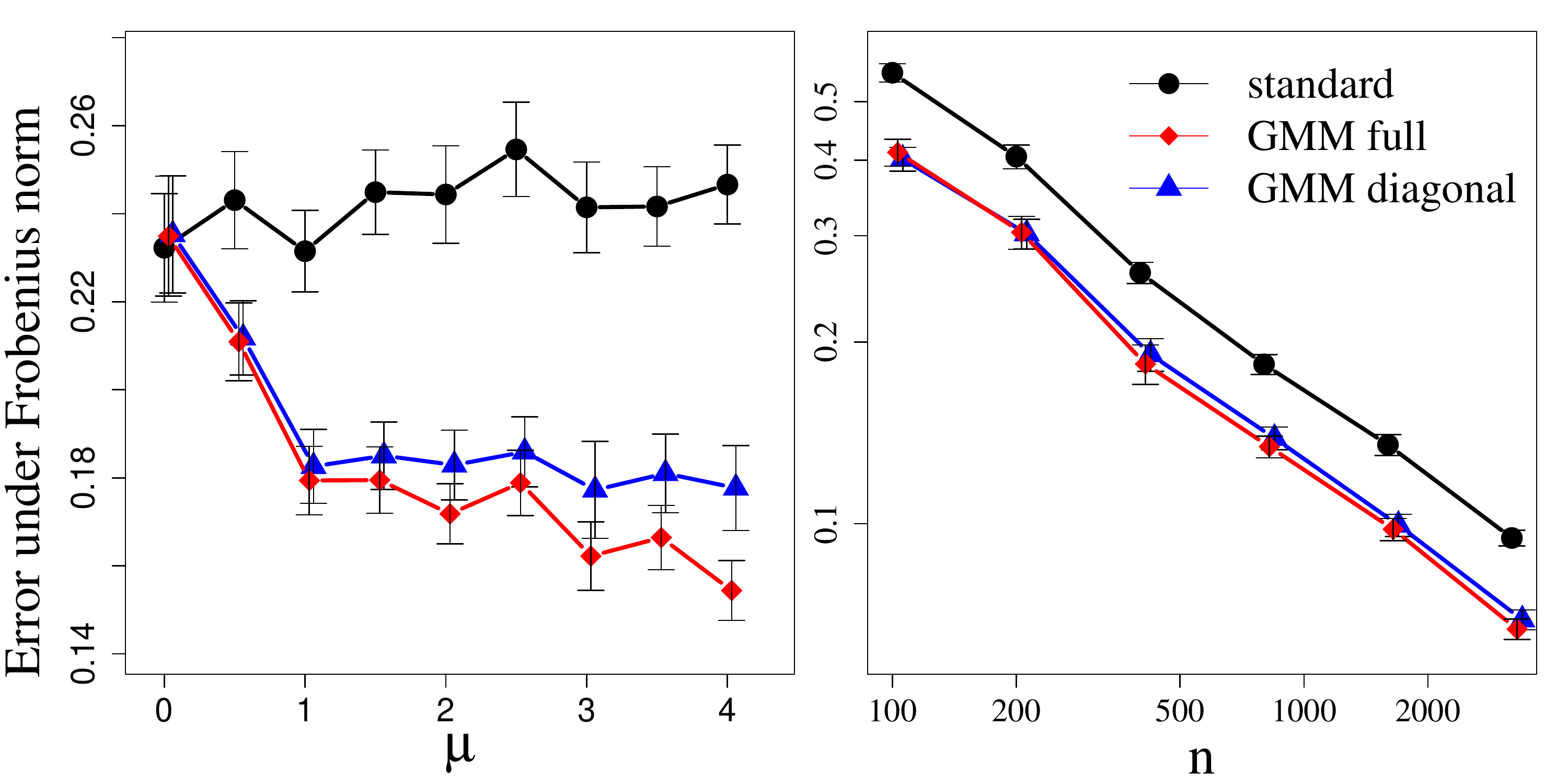}
    \caption{Estimation errors $\E \| \hat \bU\hat \bU^T - \bU^* (\bU^*)^T \|_F$ of three estimators calculated based on averages over $100$ simulations. On the markers, bars represent standard errors of the averages. {\bf Left plot}: fix $n = 500$.  {\bf Right plot} (on the log scale): fix $\mu = 2$.}\label{fig:ex1}
\end{figure*}

On the left plot of Figure~\ref{fig:ex1}, as $\mu$ increases, we can see that the standard method does not benefit from $\mu_\bz$'s deviation from $\bzero$, whereas both of the two GMM methods have decreasing estimation errors. Note that when $\mu = 0$, there is no information we can exploit from $\bv_1$ (first moments) since its expectation is zero, so we expect all methods behave similarly. Indeed, the figure supports this prediction at $\mu = 0$. However, as $\mu$ starts increasing, both GMM estimators have superior performance, which suggests that $\bv_1$ indeed contributes to the improvement on the standard method.  This contribution grows as $\mu$ continues to deviate from $0$. Thus, the inclusion of first moments through our GMM estimation helps if the factors do not have zero mean.

The right plot of Figure~\ref{fig:ex1} shows the same estimation error against varying $n$ on the logarithmic scale for all estimators. All three curves have the same alignment, which is consistent with the theory that errors scale with $n^{-1/2}$. Moreover, both GMM estimators outperform the standard one regardless of $n$, which further demonstrates the benefits of first moments.

\subsection*{Example 2: advantages of transformations}

We now look at another appealing aspect of our GMM methods: using different transformations in the construction of $\bv_\ell$, we may expose and extract even more information from GMM estimators. Consider the mixed linear regression model \eqref{model:linreg} with identical $\sigma_k$.

We fix $p=10$ and $K=2$, and sample each $\bbeta_k$ independently and uniformly from a sphere in $\R^p$ with radius $4$. We also generate $\beta_{k0} \sim N(0,1)$. Then, we generate i.i.d.\ entries of $\xx$ from $N(0,1)$, i.i.d.\ $z_i \sim \mbox{Bernoulli}(1/2)$ and i.i.d.\  $\epsilon_i \sim N(0,\sigma^2)$, where $\sigma$ is fixed as $1$. The left singular vectors of $[ \bbeta_1, \bbeta_2]$ are set to $\bU^*$. We collect $m=25$ vectors $\bv_\ell$ as below.

\begin{itemize}
\item[(a)] first moments: $\bv_1 = n^{-1} \sum_{i=1}^n y_i \xx_i; $
\item[(b)] second moments: $\bv_{1+j} = n^{-1}\sum_{i=1}^n y_i ^2\left( \xx_i \xx_i^T \be_j - \be_j\right), \quad \forall \, j \in[p]; $
\item[(c)] transformed first moments: $\bv_{11+j} = n^{-1} \sum_{i=1}^n \cos\left(y_i \pi/(2\tau) + (j-1)\pi/4\right) \xx_i, ~ \forall \, j \in[4]; $
\item[(d)] transformed second moments: $\bv_{15+j} = n^{-1}\sum_{i=1}^n \tilde y_i \left( \xx_i \xx_i^T \be_j - \be_j\right), \quad \forall \, j \in[p]. $
\end{itemize}
In (c), $\tau$ is a scale parameter chosen to be $0.8$ quantile of $| y_i |$; and in (d), the transformed response $\tilde y_i$ is defined by
\begin{equation*}
\tilde y_i := \sgn(y_i)\, \sgn(\xx_i^T \tilde \bv_1), \quad \text{where }\tilde \bv_1: = n^{-1} \sum_{i=1}^n \sgn(y_i) \xx_i.
\end{equation*}

\begin{figure*}[t!]
    \centering
    \includegraphics[scale=0.45]{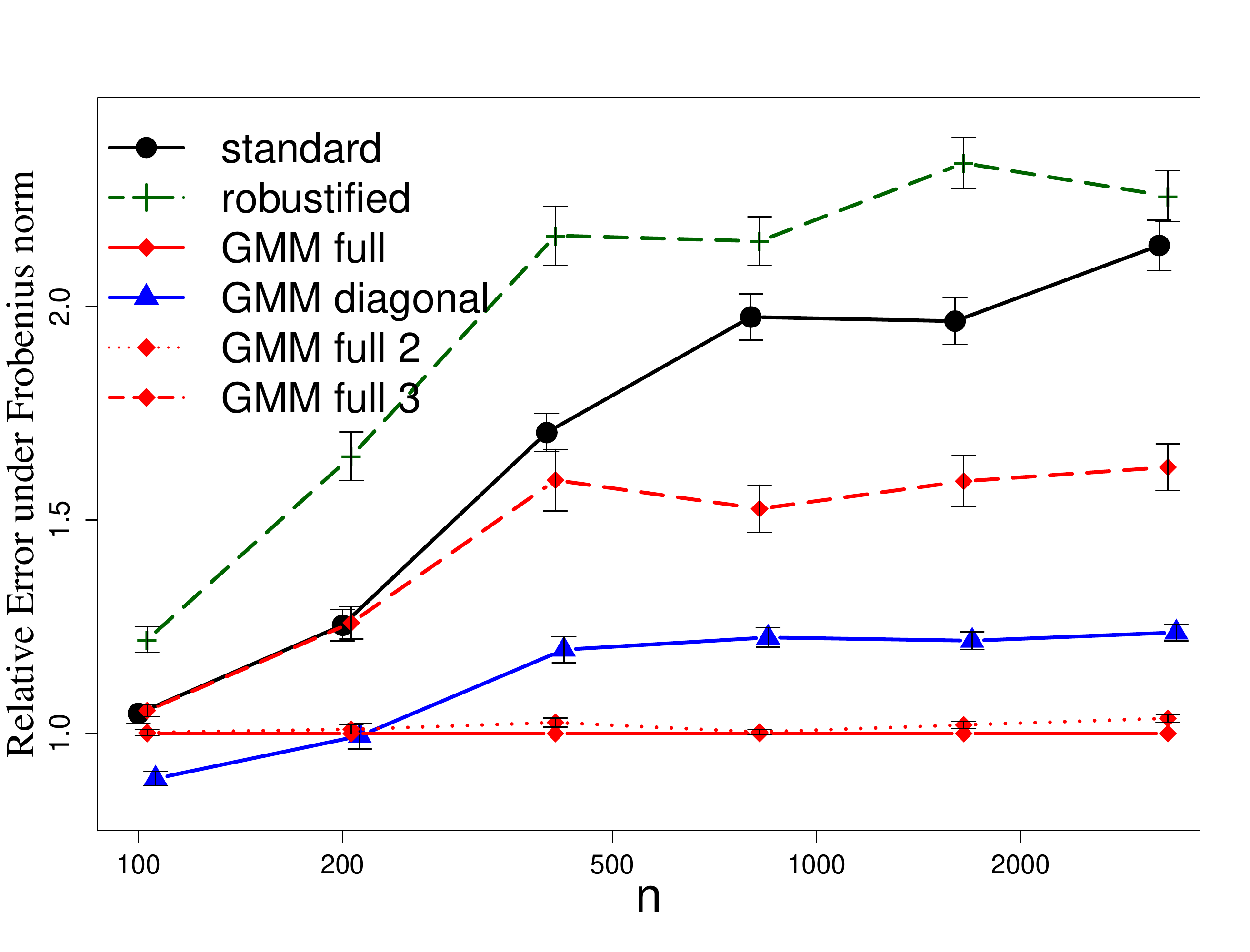}
    \caption{Ratios of estimation errors $\E \| \hat \bU\hat \bU^T - \bU^* (\bU^*)^T \|_F$ calculated based on averages over $100$ simulations. Five methods (five curves) are compared against `GMM full' (the horizontal line $y=1$). On the markers, bars are shown to represent standard errors. \label{fig:ex2}}
\end{figure*}

In essence, (d) seeks to robustify the second moments by
only using the sign of $y_i$. The additional $\sgn(\xx_i^T \tilde \bv_1)$ avoids vanishing $\E \bv_{15+j}$, which is proposed by \cite{SunIoaMon14}. Note that (d) does not strictly satisfy the condition \eqref{def:v2}, because the sign function is not smooth and $\tilde \bv_1$ is only asymptotically a linear combination of $\bbeta_1$ and $\bbeta_2$. Nevertheless, we find it practically useful for our model here.

Gauged by the same error $\E \| \hat \bU\hat \bU^T - \bU^* (\bU^*)^T \|_F$, we compare six estimators as listed in Table~\ref{tab:1}.
\begin{table}[t!]
\begin{tabular}{ll}
\hline
\textit{standard}: only using (b) & \textit{robustified}: only using (d) \\
\textit{GMM full}: combining (a)--(d) & \textit{GMM diagonal}: combining (a)--(d), but $\bW$ is diagonal\\
\textit{GMM full 2}: combining (a),(b),(d) & \textit{GMM full 3}: combining (a),(b).\\
\hline
\end{tabular}
\caption{Comparing six methods: two atomic methods and four GMM-based methods. }\label{tab:1}
\end{table}
In Figure~\ref{fig:ex2}, we plot the performance of these estimator averaged over $100$ simulations. We make a few observations:
\begin{itemize}
\item Atomic methods like `standard' or `robustified' have higher errors than GMM methods, suggesting it is better to combine different moments.
\item `GMM full 3' has better performance than `standard', which reinforces our conclusion of the previous example that `first moments may help'.
\item `GMM full' and `GMM full 2' have further improved performance, which indicates that (d) contains useful information. Combining all moments as `GMM full' is the best.
\end{itemize}

A minor observation is that, when $n$ is small, `GMM diagonal' may be preferred over `GMM full'. It is likely because estimating the full matrix $\bW$ is difficult with a small sample size. To conclude, by constructing $\bv_\ell$ with different transformations, we may exploit and combine different information using our GMM-based methods.

\section{Real data examples}\label{sec:real}

We apply our procedure to the \textit{ozone} dataset. This dataset has been studied in \cite{BreFri85} and \cite{Li92}, etc. It contains $n=330$ days of measurements of ozone concentration and $p=8$ meteorological features, where all variables take continuous values.  Dataset details can be found in \cite{BreFri85}. The goal is to study the relationship between the ozone concentration (response) with meteorological features (predictors). We will use this dataset to study different methods for the multiple index model.

The \textit{abalone} dataset is also studied. Results and details of the dataset are in the supplementary materials. In our experiments, we use the same procedure to produce our GMM estimators as in the simulations. We also fix the parameter $\delta_n$ to be $0.01$ as before.


We obtain the ozone dataset from the R package `gclus' \citep{Hurley}. First, we run a least squares fit of the ozone concentration $y$ against all $p=8$ standardized covariates, which results in an R-squared $0.69$.  This serves as a baseline for our subsequence comparisons.

Next, we consider three methods: residual-based pHd, GMM diagonal and GMM full (recall definitions in Sect.~\ref{sec:sim}, Example~3). Before running these methods, we make a linear transformation of the data such that the covariates have identity sample covariance matrix (a.k.a.\ whitening). Then, we let $K = 2$, and run each of these three methods to find $2$ orthogonal directions, say, $\hat \bu_1$ and $\hat \bu_2$, which allows us to reduce the number of covariates to $8$ to $2$ (the new covariates have the form $\xx^T \hat \bu_k$). Finally, as in \cite{Li92}, we fit a quadratic regression. We include cross products between variables in the quadratic regression. Using the same procedure, we also run these methods with $K=1$ and $K=3$. The values of R-squared are summarized in Table~\ref{tab:2}. The adjusted R-squared values are very similar, which are omitted here.
The outperformance of our aggregated GMM methods can be easily seen.

\begin{table}[t!]
\begin{center}
\begin{tabular}{c|c|c|c|c}
\hline
\hline
 & least squares & $r$-based pHd & GMM diagonal & GMM full \\
\hline
$K=1$ &  & $0.67$ & $0.74$ & $0.74$ \\
\cline{1-1} \cline{3-5}
$K=2$ & $0.69$ & $0.69$ & $0.76$ & $0.76$ \\
\cline{1-1} \cline{3-5}
$K=3$ &  & $0.72$ & $0.77$ & $0.77$\\
\hline
\hline
\end{tabular}
\end{center}
\caption{R-squared values of the four methods under three settings. A quadratic regression is fit after $\hat \bu_k$ is obtained. In the first setting $K=1$, we find a single direction $\hat \bu_1$; in the second setting $K=2$ and the third setting $K=3$, we find $\hat \bu_k$ and run a quadratic regression including cross products between the two new covariates $\xx_i^T \hat \bu_k$.}\label{tab:2}
\end{table}

To understand the variability of our results, we sample $100$ bootstrap samples and run the same methods. In Figure~\ref{fig:realdata1} (see supplementary materials), we report the boxplots from $100$ bootstrap samples, which show the R-squared values of quadratic fits with $K=2$.

From Table~\ref{tab:2} and Figure~\ref{fig:realdata1}, we can see that both GMM methods lead to better regression fits than the naive least squares and the pHd method. Moreover, for the pHd method, the outliers of the boxplot suggest that there are failure chances, that is, pHd method does not correctly find good direction $\hat \bu$, whereas, our GMM methods show robust performance in all subsamples. This is consistent with our findings in Example~3 of Section~\ref{sec:sim}.

\section{Discussions}\label{sec:disc}

In this paper, we proposed an estimation framework via GMM to combine information from overidentified vectors and estimate an unknown subspace. This approach is applied to a variety of statistical problems.

A natural question that may be explored is to allow $p$ or $m$ to grow with $n$. The large $p$ regime is relevant in the high dimensional literature, in particular, in the presence of sparsity or matrix incoherence structure, e.g.\ sparse principal component analysis \citep{ZouHasTib06}, matrix completion \citep{CanRec09}, etc. The large $m$ regime is related to moments selection in the GMM literature, which studies selecting informative moments from a large pool of candidate moment conditions (that is, vectors $\bv_\ell$ in our problem).

Also, it is interesting to see whether our method may help modern problems, such as (nonlinear) matrix completion and neural nets. Similar to the pHd, we might consider combining different transformations and activation functions for optimal results. For example, \cite{CohSha16} considered tensor structure for convolutional networks, and \cite{Wan18} studied data dependent link function (or activation function) for deep neural nets.




\appendix
\newpage

\bigskip
\begin{center}
{\large\bf SUPPLEMENTARY MATERIAL}
\end{center}

This supplementary document consists of two additional simulation experiments, one additional real data example, proofs, and other details.

\section*{A. Additional simulation and data examples}

\subsection*{Example 3: a remedy for vanishing moments}

In this example, we study multiple index models in Section~\ref{sec:indexmodel}. Related to Example~1 and~2, we illustrate how GMM methods are able to avoid the issue of vanishing moments by collecting information from different moments. We focus on three specific forms:
\begin{align*}
y_i &= \cos(2 \xx_i^T \bbeta_1) - \sin(\xx_i^T \bbeta_2) + 0.5\epsilon_i, \label{model:a} \tag{Model A} \\
y_i &= \cos(2 \xx_i^T \bbeta_1) - \xx_i^T \bbeta_2 + 0.5\epsilon_i, \label{model:b} \tag{Model B} \\
y_i &= \cos(2 \xx_i^T \bbeta_1) - \cos(\xx_i^T \bbeta_2) + 0.5\epsilon_i. \label{model:c} \tag{Model C}
\end{align*}

\begin{figure*}[b!]
    \centering
    \includegraphics[scale=0.45]{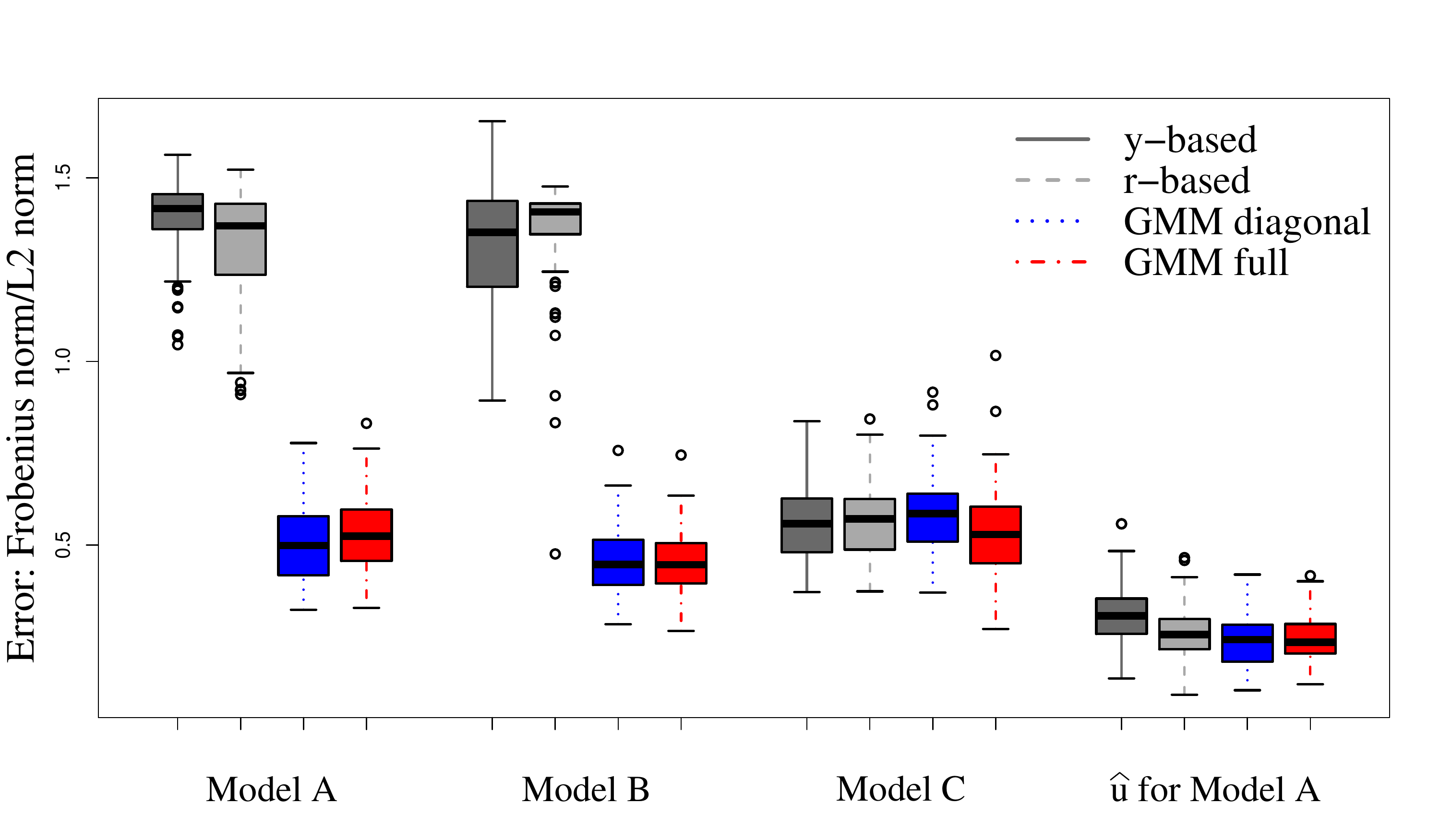}
    \caption{{\bf First three groups}: Distribution of $\| \hat \bU\hat \bU^T - \bU^* (\bU^*)^T \|_F$ for four methods under three models from $100$ simulations. {\bf Last group}: Distribution of $\| \bP_{\bU^*}^\bot \hat \bu \|_2$ under Model A from $100$ simulations.}\label{fig:ex3}
\end{figure*}

Model C has been considered by \cite{Li92}. Here we adopt the same parameter setup as \cite{Li92}. Let $n=400$, $p=10$, and $r=2$ (since there are only two $\bbeta_k$s). We fix $\bbeta_1 = \be_1 = (1,0,0,\ldots)^T$ and $\bbeta_2 = \be_2 = (0,1,0,\ldots)^T$, and generate i.i.d.\ $\xx_i \sim N(0, \bI_p)$ and $\epsilon_i \sim N(0,1)$. Set $\bU^* = \spann(\bbeta_1, \bbeta_2)$.

We compare four methods for these three models: $y$-based pHd method, $r$-based pHd method, GMM diagonal and GMM full. The first two methods are from \cite{Li92}, and the last two are the same as in Example 1. For the GMM methods, we use first moments and transformed first moments as in Example 2 (see (a) and (c)), as well as the moments constructed from both pHd methods:
\begin{equation*}
n^{-1}\sum_{i=1}^n y_i \left( \xx_i \xx_i^T \be_j - \be_j\right), \quad n^{-1}\sum_{i=1}^n r_i \left( \xx_i \xx_i^T \be_j - \be_j\right), \quad \forall\, j \in[p].
\end{equation*}
As \cite{Li92}, we center the data $\xx_i$ and $y_i$  first (but this is not essential in our case). We compute the same estimation errors for all methods and all models over $100$ simulations. In addition, for the top eigenvector $\hat \bu$ produced by all four methods, we compute $\E \| \bP_{\bU^*}^\bot \hat \bu \|_2$, which is the expected amount of the part of $\hat \bu$ unexplained by $\cS$.

In Figure~\ref{fig:ex3}, the first three groups of boxplots show the distribution of $\| \hat \bU\hat \bU^T - \bU^* (\bU^*)^T \|_F$, and the last group shows that of $\| \bP_{(\bU^*)^\bot} \hat \bu \|_2$. Clearly, the GMM methods are much better than the pHd methods for Model~A and~B, due to the fact that the pHd methods tend to miss linear trend (or more generally odd functions). On model C, all methods perform similarly, since pHd methods capture most useful information from second moments.

From the last group in Figure~\ref{fig:ex3}, we also observe that the quality of the first eigenvector $\hat \bu$ produced by pHd methods is roughly on a par with the GMM methods. This suggests that, pHd methods can, after all, find a first direction, though they fail to find a second (missing the direction of $\bbeta_2$). This is where GMM methods are very effective, because they can find directions from various moments and collect them all.

\subsection*{Estimation of Subspace Dimension}

Lastly, we study the estimation of $r$. We consider the same factor model as in Example 1. As before, we fix $p=10$, $\mu = 2$, $\sigma = 2$. We consider two cases: $r = 2$ and $r=4$, with $\mu_\bz = (\mu,-\mu)^T$ or $\mu_\bz = (\mu,-\mu,\mu,-\mu)^T$. For both cases, we set $n = 250\,r$. The mechanism for sampling parameters and generating random variables remain the same.

\begin{figure*}[t!]
    \centering
    \includegraphics[scale=0.45]{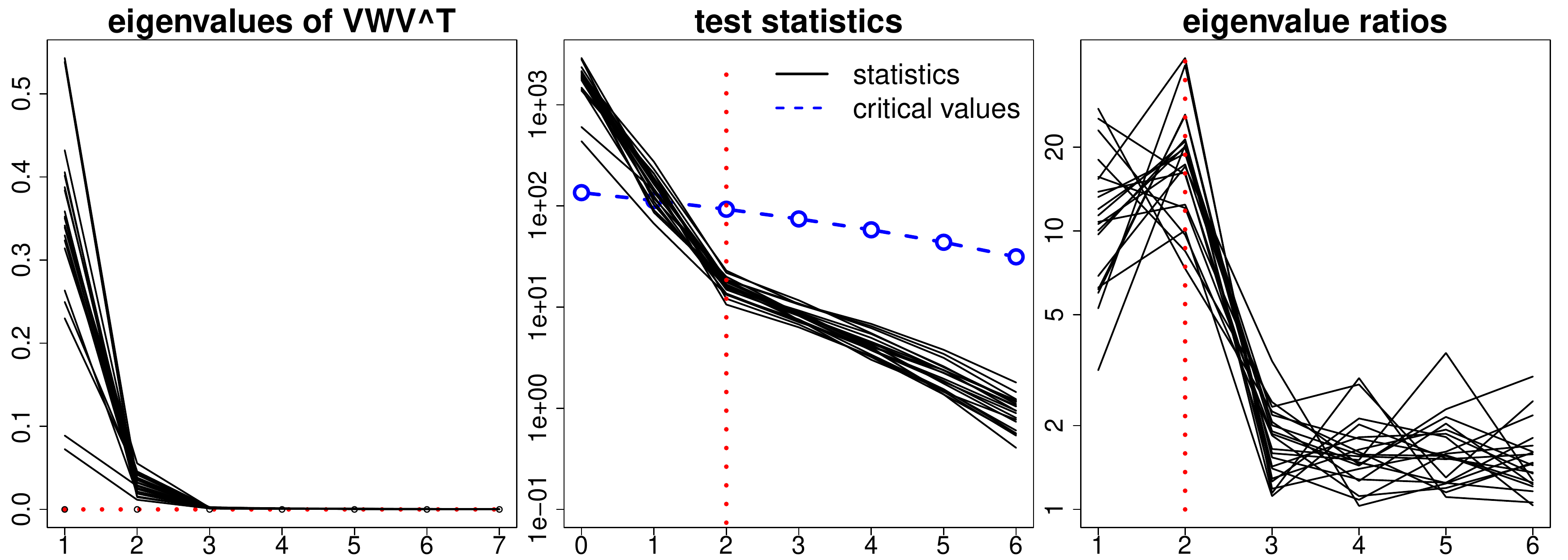}
    \caption{$r=2$. We plot $20$ solid curves computed from $20$ simulations. The $y$-axis is on the log scale on the middle plot and the right plot.}\label{fig:dimest1}
\end{figure*}

For each case, we make three plots that correspond to different methods (see Figure~\ref{fig:dimest1} and~\ref{fig:dimest2}) from $20$ simulations. Let $\lambda_k := \lambda_k(\bV \bW \bV^T)$ be the $k$th largest eigenvalue computed from the GMM methods.
\begin{itemize}
\item The left plots show $\lambda_k$ with different $k$.
\item The middle plots show $n(p-k) \sum_{j > k} \lambda_j$ with different $k$. The dashed curve with circle markers plots the critical value, namely $95\%$ quantile of $\chi^2_{(p-k)(m-k)}$, for different $k$. Note that $k$ starts from $0$.
\item The right plots show the eigen-ratio $\lambda_k / \lambda_{k+1}$ with different $k$.
\end{itemize}

Moreover, we add dotted lines on each plots. The horizontal dotted lines on the left plots have zero values on the y-axis. If some $\lambda_j$ is close to this line, then the dimension $r$ should be smaller than $j$. Indeed, in both cases, the eigenvalue curves do not visibly touch the dotted lines until $j = r+1$. This leads to good dimension estimator $\hat r_\tau$ with any reasonable parameter $\tau_n$.

On the middle and right plots, the vertical dotted lines represent the true dimension $r$. The method based on the chi-squared asymptotics \eqref{def:rhateta} mostly likely leads to $\hat r = r$ or $\hat r = r-1$, which may underestimate $r$. This is probably due to the fact that the  chi-squared asymptotics depend on stronger assumptions (see Theorem~\ref{thm:dimest}). For both eigen-ratio curves, we see a large spike at the true $r$ in most simulations, which suggests they are useful in finding $r$.

\begin{figure*}[t!]
    \centering
    \includegraphics[scale=0.45]{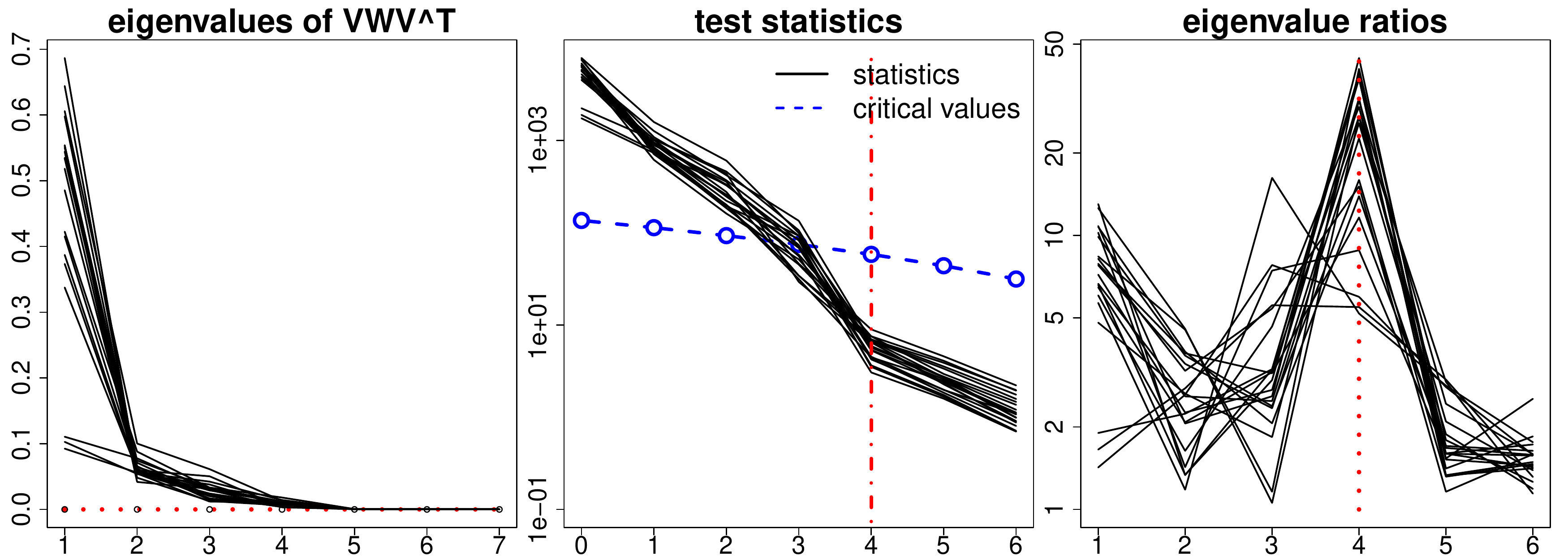}
    \caption{$r=4$. Solid curves are computed from $20$ simulations. The $y$-axis is on the log scale on the middle plot and the right plot.}\label{fig:dimest2}
\end{figure*}

In practice, we can consider combining all these plots to choose $r$. More sophisticated methods, such as those based on information criterion \citep{BaiNg02}, may be helpful, but they are out of the scope of this paper.

\subsection*{The abalone dataset}

The abalone dataset is a popular one from UCI machine learning repository \citep{Dua:2017}. It contains $n=4177$ sets of measurements of abalone, where each set of measurements consists of $7$ physical quantities with continuous values, and one discrete variable indicating the sex (male, female or infant). The goal is to predict the rings (or equivalently, the ages) of abalone. For our illustration, we will treat the sex as an unobserved (latent) variable, and study a mixed linear regression model with different methods.

We consider the mixed linear regression model \eqref{model:linreg} for our second data example. Each data unit has a `sex' variable $z$, which can only take three values (male, female and infant). According to the value of this variable, we group our dataset into three sub-datasets, and set $K=3$. We observe strong correlations between the $7$ physical variables (the minimum correlation is $0.77$), so we decide to only use their principal components. We fix $p=5$ and treat the $5$ principal components as our covariates.

For each sub-dataset, we run a least squares regression of the `rings' variable $y$ against the covariates. The resulting coefficient vectors $\bbeta_1,\bbeta_2,\bbeta_3 \in \R^p$ (without the intercepts) are treated as the true parameters, and our goal is to estimate $\spann\{\bbeta_1,\bbeta_2,\bbeta_3\}$. Henceforth, we treat $z$ as a latent variable.

We run and compare three methods: standard (second moments), GMM diagonal and GMM full. The explanations of the three methods are in Table~\ref{tab:1}. Under the same Frobenius norm as considered in Section~\ref{sec:sim}, the three methods give an error of $1.52$, $1.44$ and $0.78$ respectively. We also implement the same procedure for $100$ bootstrap subsamples and make boxplots of the errors. The results are shown in Figure~\ref{fig:realdata2}.

\begin{figure*}[t!]
    \centering
    \includegraphics[scale=0.45]{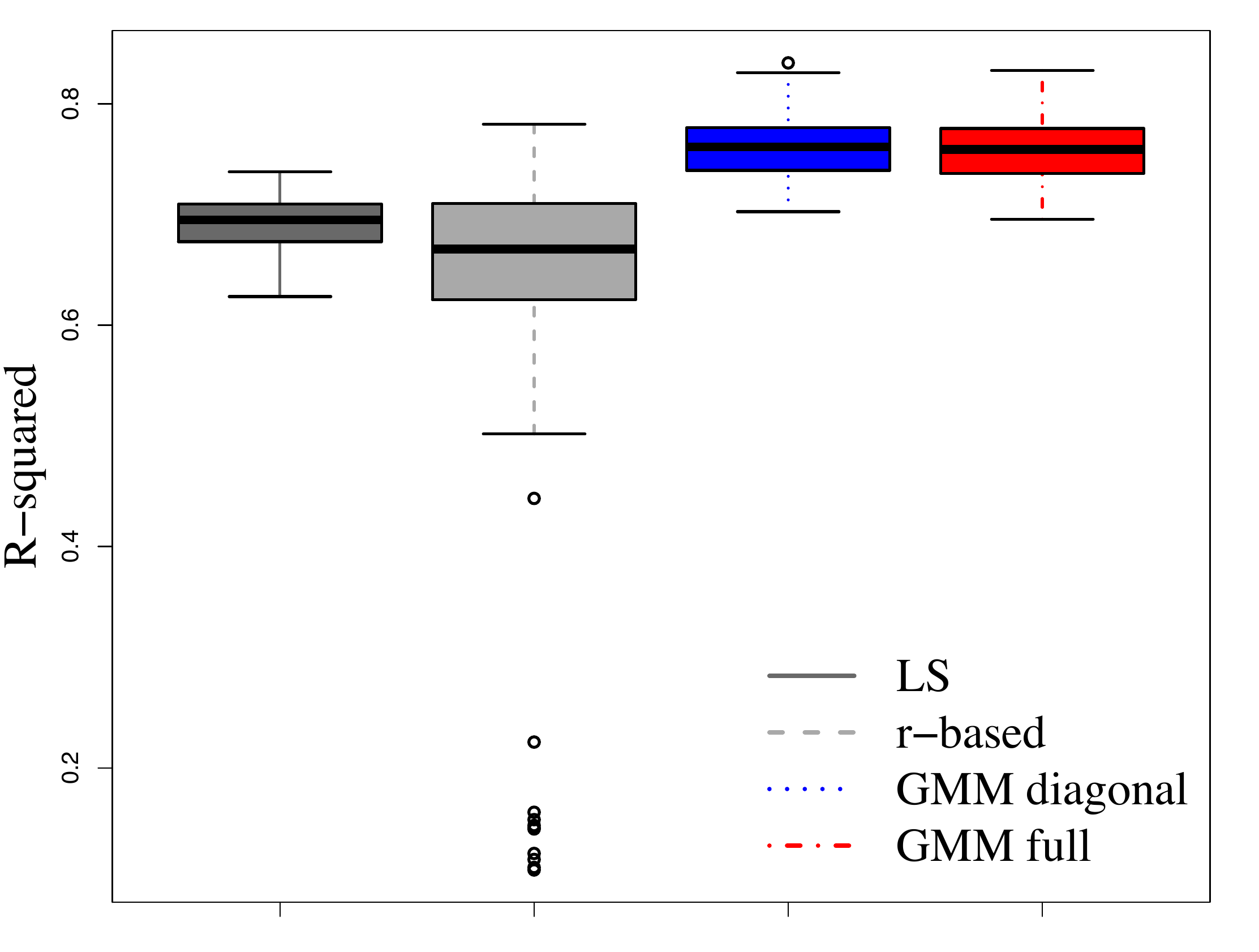}
    \caption{Comparison of the four methods based on $100$ bootstrap samples. The R-squared values are calculated after fitting a quadratic regression (including cross products).}\label{fig:realdata1}
\end{figure*}

\begin{figure*}[t!]
    \centering
    \includegraphics[scale=0.45]{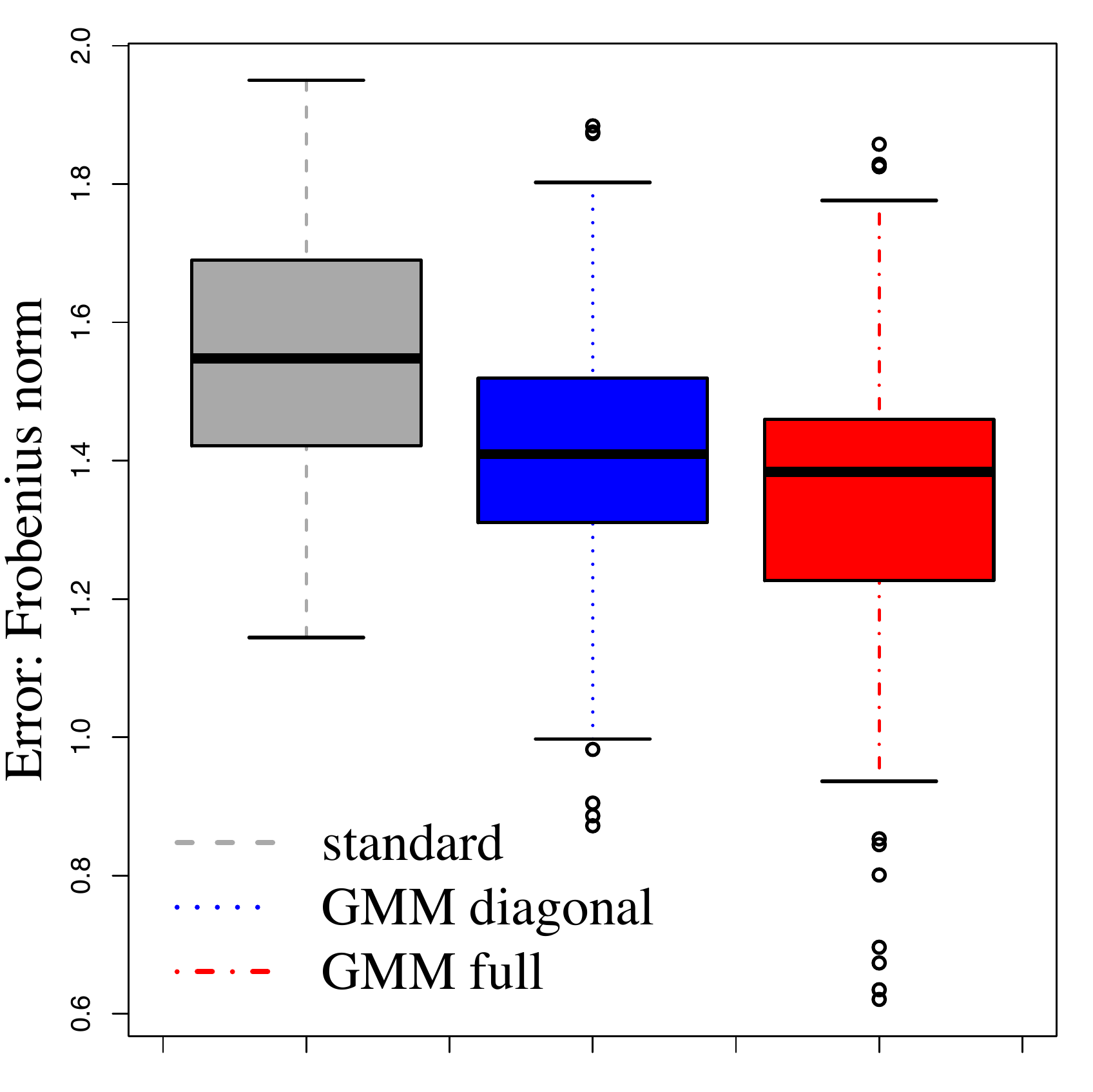}
    \caption{Comparison of the three methods from $100$ bootstrap subsamples. The boxplots show the distributions of $\| \hat \bU\hat \bU^T - \bU^* (\bU^*)^T \|_F$. }\label{fig:realdata2}
\end{figure*}

As shown by the results, our two GMM estimators outperform the standard second moments. Note that this dataset is a challenging one, since the standard method produces an estimation error that exceeds $1$. Even so, the GMM estimators have improvements over the standard method, which shows the robustness of our GMM estimators.

\section*{B. Proofs}

\subsection*{Technical lemmas}
The next lemma lists a few useful properties of the Kronecker product.
\begin{lem}\label{lem:Kron}
Suppose $\bA,\bB,\bC,\bD$ are matrices with specified dimensions, or with appropriate dimensions such that matrix products are valid. Then, we have the following identities:  \\
(i) $(\bA \otimes \bB) ( \bC \otimes \bD) = (\bA \bC) \otimes (\bB \bD)$; \\
(ii) $( \bA \otimes \bB)^{-1} = \bA^{-1} \otimes \bB^{-1}$; ~ (assuming both $\bA$ and $\bB$ are invertible) \\
(iii) $( \bA \otimes \bB)^{T} = \bA^{T} \otimes \bB^{T}$;\\
(iv) $\rank(\bA \otimes \bB) = \rank(\bA)\, \rank(\bB)$; \\
(v) If $\bA, \bB$ are both orthogonal projection matrices, then $\bA \otimes \bB$ is also an orthogonal projection;\\
(vi) $\vect(\bA \bB \bC) = (\bC^T \otimes \bA) \vect(\bB)$; ~ (recall $\vect(\cdot)$ means vectorizing a matrix column-wise)\\
(vii) There is a unique permutation matrix $\bP(p,m)$, depending only on dimension $p$ and $m$, such that $\vect(\bA^T) = \bP(p,m) \vect(\bA)$ for any $\bA \in \R^{p \times m}$;\\
(viii) $\bA \otimes \bB = \bP(p,m) (\bB \otimes \bA) \bP(p,m)^T$ holds for any $\bA \in \R^{p \times p}$, $\bB \in \R^{m \times m}$, where $\bP(p,m)$ is given in (vii); \\
(i$\mspace{1mu}\mathrm{x}$) $\Tr(\bA^T \bB \bA) = (\vect(\bA))^T ( \bI_p \otimes \bB) \vect(\bA)$, where $p$ is the number of columns of $\bA$.
\end{lem}
The proof of these properties is known: (i)--(iv) are proved in Sect.\ 4.2 of \cite{Hor91}, (vi)--(viii) are proved in Sect.\ 4.3 of the same book, and (v) is straightforward from the definition. We give a short proof of (i$\mspace{1mu} \mathrm{x}$) below.
\begin{proof}[\bf Proof of Lemma~\ref{lem:Kron} (i$\mspace{1mu}\mathrm{\bx}$)]
Let us write $\bA = [\ba_1,\ldots,\ba_p]$, where $\ba_j \in \R^p$. Using this notation, we have $\Tr(\bA^T \bB \bA) = \sum_j \ba_j^T \bB \ba_j$. Since $\bI_p \otimes \bB$ is a block diagonal matrix, we also have $(\vect(\bA))^T ( \bI_p \otimes \bB) \vect(\bA) = \sum_j \ba_j^T \bB \ba_j$, and thus we obtain the desired identity.
\end{proof}
We will also use an eigenvalue perturbation result from \cite{Kato66} (see Chap.\ 2, p.\ 79, Eq. (2.33)). Here we present a simplified form used in \cite{Li91} and \cite{Li92}.

\begin{lem}\label{lem:pert}
Consider the asymptotic expansion
\begin{equation*}
\bT(\omega) = \bT + \omega \bT^{(1)} +  \omega^2 \bT^{(2)} + o(\omega^2),
\end{equation*}
where $\omega = o(1)$, and $\bT(\omega), \bT, \bT^{(1)},\bT^{(2)} \in \R^{p \times p }$ are symmetric matrices. Suppose that $\bT$ has rank $k$, where $k < p$. Let $\lambda(\omega)$ be the sum of $p-k$ eigenvalues of $\bT(\omega)$ with smallest absolute values. Let $\bPi \in \R^{p \times p}$ be the projection matrix associate with the null space of $\bT$ so that $\bPi \bT = \bT \bPi = \bzero$. Then,
\begin{equation*}
\lambda(\omega) = \omega \lambda^{(1)} + \omega^2 \lambda^{(2)} +  o(\omega^2),
\end{equation*}
where $\lambda^{(1)} = \Tr( \bT^{(1)} \bPi)$ and $ \lambda^{(2)} = \Tr(\bT^{(2)} \bPi - \bT^{(1)} \bT^+ \bT^{(1)} \bPi) $. Here we use $^+$ to denote the Moore-Penrose pseudoinverse.
\end{lem}

The next lemma gives a useful elementary property about the Moore-Penrose pseudoinverse. Recall that for any matrix $\bA \in \R^{n_1 \times n_2}$, the Moore-Penrose pseudoinverse of $\bA$ is denoted by $\bA^+$.
\begin{lem}\label{lem:pseudo}
The matrix $\bA (\bA^T \bA)^+ \bA^T$ is an orthogonal projection matrix for any $\bA$, and its rank is equal to the rank of $\bA$.
\end{lem}
\begin{proof}[Proof of Lemma~\ref{lem:pseudo}]
It can be shown that $(\bA^T \bA)^+ \bA^T = \bA^+$ for any $\bA$ (see \cite{Ben03} Chap.\ 1, p.\ 49, Ex.\ 18). Furthermore, it is known that  $\bA \bA^+$ is an orthogonal projection onto the range of $\bA$ \citep[Chap.\ 5.5.2]{Gol13}. Thus, the conclusion follows.
\end{proof}

\subsection*{Proofs for Section~\ref{sec:normal}}

\begin{proof}[\bf Proof of Lemma~\ref{lem:sigmaS}]
Denote $\bP = \bI_p - \bU^* (\bU^*)^T$. It is clear that $\bP$ is a deterministic projection matrix with rank $p-r$, so $\Tr(\bP) = p-r$. For simplicity, let us also denote $\bff_j = \bff_j(\xx_i, y_i)$ and $\bar \bff_j = \E \bff_j$. For any $j,\ell \in [m]$, by definition and linearity of $\Tr(\cdot)$ and $\E(\cdot)$,
\begin{align*}
\Sigma^*_{j\ell} &= \Tr\left(  \E  \left[   \bff_j^T \bP\,  \bff_{\ell}    \right] \right)  = \E\left[	 \Tr\left( (\bff_j - \bar \bff_j)^T \bP\,  (\bff_{\ell} - \bar \bff_\ell)  	\right)\right] \\
&= \E\left[	 \Tr\left(  \bP\,  (\bff_{\ell} - \bar \bff_\ell)  (\bff_j - \bar \bff_j)^T	\right)\right] \\
&= \Tr\left(	 \E\left[  \bP\,  (\bff_{\ell} - \bar \bff_\ell)  (\bff_j - \bar \bff_j)^T	\right]\right) \\
&= \Tr\left( \bP\,  S_{\ell j} \bI_p\right) \\
&= (p-r)S_{j \ell}.
\end{align*}
In the above derivation, we used the fact that $\bP \bar \bff_j = \bzero$, and that $\E  (\bff_{\ell} - \bar \bff_\ell)  (\bff_j - \bar \bff_j)^T$ is the cross-covariance between $\bff_\ell$ and $\bff_j$, and thus a submatrix of $\bar \bS$.
\end{proof}

We state and prove a theorem that is more general than Theorem~\ref{thm:consist2}.

\begin{thm}\label{thm:consist2}
Suppose Assumption~\ref{ass:span} holds and $\bW_n \xrightarrow{p} \bW^* \succeq \bzero$, where $\bW^* \in \R^{m \times m}$ satisfies $(\bG^*)^T \bW^* \bG^* \succ \bzero$. Then, there exists a sequence of orthogonal matrices  $\bR_1, \bR_2, \bR_3, \ldots \in O(r)$ such that
\begin{equation*}
\hat \bU_n \bR_n \xrightarrow{p} \bU^*, \qquad \text{as} ~ n \to \infty.
\end{equation*}
\end{thm}

\begin{proof}[\bf Proof of Theorem~\ref{thm:consist2}]
Fix any $\delta > 0$ independent of $n$. As stated, we suppress the subscript $n$. Let $N_{\delta}(\bU^*)$ be a neighborhood of $\bU^*$, up to rotation:
\begin{align*}
N_{\delta}(\bU^*) &= \{\bU \in O(p,r): \exists\, \bR \in O(r), \| \bU - \bU^* \bR \|_F < \delta \} \\
&= \cup_{\bR \in O(r)} \{ \bU \in O(p,r): \| \bU - \bU^* \bR \|_F < \delta \},
\end{align*}
which is an open set in $O(p,r)$. Recall $\bar \bW^* = \bW^* \otimes \bI_p$ and $Q^*(\bU) = [ \E \bg(\bU)]^T \bar \bW^* [\E\bg(\bU)]$. Also define $\veps := \inf_{\bU \notin N_{\delta}(\bU^*)} Q^*(\bU)$. Note that $Q^*(\bU)$ and $\veps$ do not depend on $n$. Also note that $ Q^*(\bU)$ is a continuous function in $O(p,r)$, and that the set $[ N_\delta (\bU^*) ]^c$, namely, the complement of $N_\delta (\bU^*)$, is compact. It follows that the minimum of $Q^*(\bU)$ over $[ N_\delta (\bU^*) ]^c$ can be attained. We claim that $\veps > 0$.

To prove this claim, recall $\bF = [\bff_1,\ldots,\bff_m]$ and $\bG^* = (\E \bF)^T \bU^*$. Since $\E (\bI_p - \bU \bU^T) \bv_\ell = (\bI_p - \bU \bU^T) \E \bff_\ell$ and $\E \bF = \bU^* (\bU^*)^T \E \bF = \bU^* (\bG^*)^T$, we can rewrite $Q^*(\bU)$ as
\begin{align*}
Q^*(\bU) &= \sum_{j,\ell = 1}^m w^*_{j\ell} ( \E \bff_j )^T (\bI_p - \bU \bU^T) \E \bff_\ell = \Tr\left( \bW^* (\E \bF)^T (\bI_p - \bU \bU^T) \E \bF\right) \\
&= \Tr \left( \bW^* \bG^* (\bU^*)^T (\bI_p - \bU \bU^T) \bU^* (\bG^*)^T \right) \\
&= \Tr \left( (\bG^*)^T \bW^* \bG^* (\bU^*)^T (\bI_p - \bU \bU^T) \bU^*  \right).
\end{align*}
Note that $ (\bU^*)^T (\bI_p - \bU \bU^T) \bU^* \succeq \bzero$, and by assumption $(\bG^*)^T \bW^* \bG^* \succ \bzero$. Thus $Q^*(\bU) = 0$ if and only if $ (\bU^*)^T (\bI_p - \bU \bU^T) \bU^* = \bzero$. Observe that
\begin{equation*}
\spann(\bU^*) = \spann(\bU) \Leftrightarrow (\bI_p - \bU \bU^T) \bU^* = \bzero \Leftrightarrow (\bU^*)^T (\bI_p - \bU \bU^T) \bU^* = \bzero.
\end{equation*}
Let $\bU_0 \in N_{\delta}(\bU^*)$ be the minimizer of $Q^*(\bU)$. It is clear that $\spann(\bU^*) \neq \spann(\bU_0)$, so from the above reasoning, we deduce $\veps = Q^*(\bU_0) > 0$.\\
Furthermore, since $\bI_p - \bU \bU^T$ is a projection matrix, for any $\bU \in O(p,r)$, we have
\begin{align*}
& \| \bg(\bU) - \E \bg(\bU) \|_2^2 = \sum_{\ell=1}^m \| (\bI_p - \bU \bU^T) ( \bv_\ell - \E \bv_\ell) \|_2^2 \le \sum_{\ell=1}^m \| \bv_\ell - \E \bv_\ell \|_2^2 \xrightarrow{p} 0, \\
& \| \bar \bW - \bar \bW^* \|_2 = \| (\bW -   \bW^*) \otimes \bI_p \|_2 \xrightarrow{p} 0.
\end{align*}
These also imply $\|\bg(\bU) \|_2$ and $\| \bar \bW \|_2$ are uniformly bounded by a constant, and therefore, uniform convergence:
\begin{equation*}
\sup_{\bU \in O(p,r)} \big| Q(\bU) - Q^*(\bU) \big| \xrightarrow{p} 0.
\end{equation*}
In particular, we have convergence in probability for $\bU = \hat \bU$ and $\bU = \bU^*$. Using the inequality $Q(\hat \bU) \le Q(\bU^*)$ (by definition of $\hat \bU$), we have, for large $n$,
\begin{equation*}
Q^*(\hat \bU) \le Q(\hat \bU) + \veps/3  \le Q(\bU^*) +\veps/3 \le Q^*(\bU^*) + 2\veps/3 < \veps.
\end{equation*}
Thus, we deduce that, for large $n$, $\hat \bU \in N_{\delta}(\bU^*)$. Since $\delta >0$ is arbitrary, we conclude that for appropriate choice of $\bR \in O(r)$, we have $\hat \bU \bR \xrightarrow{p} \bU^*$.
\end{proof}

\begin{proof}[\bf Proof of Theorem~\ref{thm:normality}]

By definition, $\hat \bU$ is a minimizer of \eqref{def:Q}. First, we establish a useful identity derived from the first-order optimality condition of $\hat \bU$. This is achieved by making use of the connection with the eigenvector formulation in \eqref{opt:eigen}.

Since the columns of $\hat \bU$ are eigenvectors of $\bV \bW \bV^T$, we must have
\begin{equation*}
\bV \bW \bV^T \hat \bU = \hat \bU \hat \bLambda, \quad \text{where} ~  \hat \bLambda := \diag\{\hat \lambda_1,\ldots, \hat \lambda_r\}.
\end{equation*}
Here, $\hat \lambda_1,\ldots, \hat \lambda_r$ are the top $r$ eigenvalues of $\bV^T \bW \bV $. It leads to
\begin{equation}\label{eq:optcond0}
\bP_{\hat \bU}^\bot \bV \bW  \bV^T \hat \bU =  \bP_{\hat \bU}^\bot\hat \bU \hat \bLambda = \bzero,
\end{equation}
Define $\bG := [\bv_1,\ldots,\bv_m]^T \hat \bU \bR \in \R^{m \times r}$. Right multiplying both sides of \eqref{eq:optcond0} by $\bR$, we obtain
\begin{equation}\label{eq:optcond}
\bP_{\hat \bU}^\bot \bV \bW  \bG = \bzero.
\end{equation}
Now we observe $  \bP_{\hat \bU}^\bot (\hat \bU \bR - \bU^*) = - \bP_{\hat \bU}^\bot \bU^*$ since $\bP_{\hat \bU}^\bot \hat \bU = \bzero$. Thus, defining $\bH := \bV^T\bU^* \in \R^{m \times r}$, we obtain
\begin{align*}
\bP_{\hat \bU}^\bot \left( \bP_{\hat \bU}^\bot \bV  - \bP_{\bU^*}^\bot \bV \right) &= - \bP_{\hat \bU}^\bot \left( \hat \bU \hat \bU^T \bV - (\bU^*)(\bU^*)^T \bV \right) = \bP_{\hat \bU}^\bot (\bU^*)(\bU^*)^T \bV \\
&= - \bP_{\hat \bU}^\bot (\hat \bU \bR - \bU^*) \bH^T.
\end{align*}
Using \eqref{eq:optcond}, we right multiply the above identity by $\bW \bG$ and obtain
\begin{equation*}
-\bP_{\hat \bU}^\bot \bP_{\bU^*}^\bot \bV \bW \bG = - \bP_{\hat \bU}^\bot (\hat \bU \bR - \bU^*) \bH^T \bW \bG.
\end{equation*}
To derive asymptotic properties for $\bP_{\hat \bU}^\bot (\hat \bU \bR - \bU^*)$,  we wish to make inversion of $\bH^T \bW \bG$ in the above equality. To that end, we make the following observation. Both $\bG$ and $\bH$ converge to the same limit in $\R^{m \times r}$, which is $\bG^* = [ \E \bv_1, \ldots, \E \bv_m]^T \bU^* \in \R^{m \times r}$, and it has full column rank under Assumption~\ref{ass:span}. By Assumption~\ref{ass:W}, $\bW \xrightarrow{p} \bW^* \succ \bzero$, so $\bG^T \bW \bH$ converges to $(\bG^*)^T \bW^* \bG^*$ in probability. These two facts also imply $(\bG^*)^T \bW^* \bG^* \succ \bzero$. Thus, with probability $1-o(1)$, $\bG^T \bW \bH$ is invertible, and its limit $(\bG^*)^T \bW^* \bG^*$ is also invertible. (Note that we have the same conclusion if we directly assumed $(\bG^*)^T \bW^* \bG^* \succ \bzero$ instead of the stronger condition $\bW^* \succ \bzero$.)

Therefore, with probability $1-o(1)$, we obtain
\begin{equation}\label{eq:useful1-0}
\bP_{\hat \bU}^\bot ( \hat \bU \bR - \bU^* )= \bP_{\hat \bU}^\bot \bP_{\bU^*}^\bot \bV \bW \bG  ( \bH^T \bW \bG)^{-1}.
\end{equation}
Note that $\vect(\bP_{\bU^*}^\bot \bV) = \bg(\bU^*)$. Using Lemma~\ref{lem:Kron} (vi), we express \eqref{eq:useful1-0} into the vector form:
\begin{equation*}
\vect( \bP_{\hat \bU}^\bot (\hat \bU \bR - \bU^*)) = ( (\bG^T \bW \bH)^{-1} \bG^T \bW \otimes \bP_{\hat \bU}^\bot )\, \bg(\bU^*)
\end{equation*}
Recall that $\bar \bS^* \in \R^{\bar m \times \bar m}$ is the covariance matrix of $ [\bff_1(\xx_i,y_i); \ldots; \bff_m(\xx_i, y_i)] \in \R^{\bar m}$ (concatenating all $\bff_\ell$ into a vector). Then, by the central limit theorem,
\begin{equation*}
\sqrt{n}\, \bg(\bU^*)  \xrightarrow{d} N(\bzero, \bar \bP_{\bU^*}^\bot \bar \bS^* \, \bar \bP_{\bU^*}^\bot), \quad \text{where}~\bar \bP_{\bU^*}^\bot = \bI_p \otimes \bP_{\bU^*}^\bot.
\end{equation*}
Note that the central limit theorem applies despite the covariance matrix of the asymptotic normal distribution being rank deficient. Moreover, since $(\bG^T \bW \bH)^{-1}\bG^T \bW\xrightarrow{p} [(\bG^*)^T \bW^* \bG^* ]^{-1} (\bG^*)^T \bW^*$ and $\bP_{\hat \bU}^\bot \xrightarrow{p} \bP_{\bU^*}^\bot$, we have
\begin{equation*}
(\bG^T \bW \bH)^{-1} \bG^T \bW \otimes \bP_{\hat \bU}^\bot\xrightarrow{p} [(\bG^*)^T \bW^* \bG^* ]^{-1} (\bG^*)^T \bW^*  \otimes \bP_{\bU^*}^\bot.
\end{equation*}

Using Slutsky's Theorem, Lemma~\ref{lem:Kron}, and the fact that $(\bP_{\bU^*}^\bot)^2 = \bP_{\bU^*}^\bot$,  we arrive at the desired asymptotic normality \eqref{asym:1}. In particular, when $\bar \bS^*$ has the block matrix form $\bar \bS^* = \bS^* \otimes \bI_p$, where $\bS^* \in \R^{m \times m}$ is invertible, then the asymptotic variance simplifies to  $(\bA \bS^* \bA^T) \otimes \bP_{\bU^*}^\bot$. From the standard GMM theory \citep{Hall05}, we deduce that the choice $\bW^* = (\bS^*)^{-1}$ is the optimal weighting matrix, in the sense that $\bA \bS^* \bA^T \succeq [(\bG^*)^T (\bS^*)^{-1} \bG^* ]^{-1}$ for any $\bW^* \succ \bzero$. It follows that \eqref{ineq:long} is true, and the equality in \eqref{ineq:long} can be attained at the choice $\bW^* = (\bS^*)^{-1}$. Note that rescaling $\bW^*$ does not change the asymptotic variance in \eqref{asym:1}. This completes the proof.

\end{proof}

\begin{proof}[\bf Proof of Theorem~\ref{thm:opt2}]
In the proof of Theorem~\ref{thm:normality}, we obtained \eqref{eq:useful1-0} under Assumption~\ref{ass:span} and \ref{ass:W}. Recall that $ \bP_{\hat \bU}^\bot (\hat \bU \bR - \bU^*) = - \bP_{\hat \bU}^\bot \bU^*$, so we have
\begin{equation*}
(\bU^*)^T \bP_{\hat \bU}^\bot \bU^* = [  \bP_{\hat \bU}^\bot (\hat \bU \bR - \bU^*)  ]^T  \bP_{\hat \bU}^\bot (\hat \bU \bR - \bU^*).
\end{equation*}
Denote $\bA :=  (\bG^T \bW \bH )^{-1} \bG^T \bW \in \R^{r \times m}$. Using \eqref{eq:useful1-0}, we have
\begin{equation}\label{eq:useful1}
(\bU^*)^T \bP_{\hat \bU}^\bot \bU^* = \bA\, \bV^T \bP_{\bU^*}^\bot\bP_{\hat \bU}^\bot \bP_{\bU^*}^\bot \bV\, \bA^T.
\end{equation}
Recall that, in Theorem~\ref{thm:normality}, we have defined $\bA^* := [(\bG^*)^T \bW^* \bG^* ]^{-1} (\bG^*)^T \bW^*$. Thus,
\begin{equation}\label{conv:AP}
\bA \xrightarrow{p} \bA^* ~~ \text{and} ~~  \bP_{\hat \bU}^\bot \xrightarrow{p} \bP_{\bU^*}^\bot.
\end{equation}
Moreover, by the central limit theorem, $\sqrt{n}\, \vect(\bP_{\bU^*}^\bot\bV) \xrightarrow{d} N(\bzero, \bar \bP_{\bU^*}^\bot \bar \bS^* \bar \bP_{\bU^*}^\bot)$ where, recall, $\bar \bS^*$ is the covariance matrix of the concatenated vector $\bar \bff_j$s. Let $\bXi = [\bxi_1,\ldots, \bxi_m] \in \R^{p \times m}$ be a random matrix such that $\vect(\bXi) \sim N(\bzero, \bar \bS^*)$. Then, $\vect(\bP_{\bU^*}^\bot \bXi) \sim N(\bzero, \bar \bP_{\bU^*}^\bot \bar \bS^* \bar \bP_{\bU^*}^\bot)$. By the (multivariate) Slutsky's theorem \citep[Thm.\ 2.7]{VanderVaart1998}, 
\begin{equation*}
n (\bU^*)^T \bP_{\hat \bU}^\bot \bU^*  \xrightarrow{d} \bA^* \bXi^T \bP_{\bU^*}^\bot \bXi ( \bA^*)^T.
\end{equation*}
We claim that $\bSigma^* = \E [\bXi^T \bP_{\bU^*}^\bot \bXi]$. In fact, for any $j,\ell \in [m]$, we have $\E [\bxi_j^T \bP_{\bU^*}^\bot \bxi_\ell] = \Tr\left( \cov[\bP_{\bU^*}^\bot \bxi_\ell, \bP_{\bU^*}^\bot \bxi_j ]  \right) $ and $\Sigma^*_{j\ell} = \E [\bff_{j}^T \bP_{\bU^*}^\bot \bff_{\ell}] = \Tr\left(  \cov[\bP_{\bU^*}^\bot \bff_\ell, \bP_{\bU^*}^\bot \bff_j ] \right)$. By definition, $\vect(\bXi)$ and $\vect(\bF)$ share the same covariance matrix $\bar \bS^*$, and therefore $\Sigma_{j\ell} = \E [\bxi_j^T \bP_{\bU^*}^\bot \bxi_\ell]$.
This leads to
\begin{equation}\label{eq:aePhi}
n \, \aE(\bPsi(\bW)) = n \, \aE((\bU^*)^T \bP_{\hat \bU}^\bot \bU^* ) = \bA^* \bSigma^* (\bA^*)^T.
\end{equation}
In particular, if $\bW$ is fixed at $\bW^*$ for all $n$, then we also have $n \, \aE(\bPsi(\bW^*)) = \bA^* \bSigma^* (\bA^*)^T$. With the choice $\bW^* = (\bSigma^*)^{-1}$, we get $\aE(\bPsi((\bSigma^*)^{-1})) = n^{-1} [(\bG^*)^T (\bSigma^*)^{-1} \bG^* ]^{-1} $. This is the smallest asymptotic expectation up to a $(1+o(1))$ factor, since for any $\bW^* \succ \bzero$,
\begin{equation}\label{ineq:optimal}
[(\bG^*)^T \bW^* \bG^* ]^{-1} (\bG^*)^T \bW^* \bSigma^* \bW^* \bG^* [(\bG^*)^T \bW^* \bG^* ]^{-1} \succeq [(\bG^*)^T (\bSigma^*)^{-1} \bG^* ]^{-1}.
\end{equation}
The above inequality is well known \citep{Hall05}. This readily implies the weighting matrix $\bW^* = (\bSigma^*)^{-1}$ is optimal. Moreover, assuming $\{ n \| (\bU^*)^T \bP_{\hat \bU}^\bot \bU^* \|_F  \}$ is uniformly integrable for $\bW^*$, we obtain
\begin{equation*}
\lim_{n \to \infty} n\, \E \big[(\bU^*)^T \bP_{\hat \bU}^\bot \bU^* \big] = n\, \aE(\bPsi(\bW^*))
\end{equation*}
by standard results (see, for example, Thm.\ 3.2.2 and Thm.\ 5.5.2 in \cite{Durrett2010}). Therefore, assuming uniform integrability, we obtain $ \E\,\bPsi((\bSigma^*)^{-1}) \preceq (1+o(1)) \E\, \bPsi(\bW^*)$. Finally, $\hat \bSigma\xrightarrow{p} \bSigma^*$, and under Assumption~\ref{ass:sigma} we have $(\hat \bSigma)^{-1} \xrightarrow{p} (\bSigma^*)^{-1}$, so the choice of $\bW = (\hat \bSigma)^{-1}$ satisfies Assumption~\ref{ass:W}, and our proof is complete.
\end{proof}

\begin{proof}[\bf Proof of Corollary~\ref{cor:procedure}]
We suppress the dependence on $\bW^*$ whenever there is no confusion. Let the singular value decomposition of $\hat \bU^T \bU^*$ be $\bU_0 \bSigma_0 \bV_0^T$, where $\bSigma_0 =  \diag\{\sigma_1,\ldots,\sigma_r\}$ is a diagonal matrix with singular values on its diagonal, and $\bU_0, \bV_0 \in O(r)$ are orthogonal matrices. Recall the definition of canonical angles $\theta_1,\ldots,\theta_r$ and $\sin \bTheta = \diag\{ \sin \theta_1,\ldots,\sin \theta_r\}$ in Section~\ref{sec:procedure}. Using $\cos \theta_k =\sigma_k$, we deduce
\begin{align*}
\Tr\big( (\bU^*)^T \bP_{\hat \bU}^\bot \bU^* \big) &= \Tr\big( \bI_r - \bV_0 \bSigma_0^T \bSigma_0 \bV_0^T \big) = r - \Tr\big(  \bSigma_0^T \bSigma_0 \big) \\
&= r - \sum_{k=1}^r \sigma_k^2 =  \sum_{k=1}^r \sin^2 \theta_k \\
&= \| \sin \bTheta \|_F^2.
\end{align*}
For any $\bW^* \succ \bzero$, by \eqref{eq:aePhi}, $n \cdot \aE( \bPsi(\bW^*))$ is the distributional limit of $n \, (\bU^*)^T \bP_{\hat \bU}^\bot \bU^*$. Taking the trace and using the above equality, we have $n \, \| \sin \bTheta(\bW^*) \|_F^2 \xrightarrow{d} n \, \Tr(\aE( \bPsi(\bW^*)))$, and thus $\aE(\| \sin \bTheta(\bW^*) \|_F^2) = \Tr(\aE( \bPsi(\bW^*)))$ up to a $1+o(1)$ factor. Since $\aE(\bPsi((\hat \bSigma)^{-1})) \preceq \aE(\bPsi(\bW^*))$ for any $\bW^* \succ \bzero$ by Theorem~\ref{thm:opt2}, we must have
\begin{equation*}
\Tr\big( \aE(\bPsi((\hat \bSigma)^{-1})) \big) \le  \Tr\big( \aE(\bPsi(\bW^*)) \big), \quad \forall \, \bW^* \succ \bzero.
\end{equation*}
Expressing this inequality equivalently in term of canonical angles, we have
\begin{equation*}
\aE(\| \sin \bTheta((\hat \bSigma)^{-1}) \|_F^2) \le \aE(\| \sin \bTheta(\bW^*) \|_F^2), \quad \forall \, \bW^* \succ \bzero.
\end{equation*}
Finally, using the equivalence \eqref{eq:mse}, we obtain the desired inequality.
\end{proof}

\subsection*{Proofs for Section~\ref{sec:dimest}}

\begin{proof}[\bf Proof of Theorem~\ref{thm:dimest}]
We use the asymptotic result in Lemma~\ref{lem:pert} to prove the theorem. Recall $\bP_{\bU^*}^\bot = \bI_p - \bU^* (\bU^*)^T$, which we write $\bP$ for shorthand. Also denote $\bV^* = \E \bV$, $\omega_n = n^{-1/2}$. We can rewrite $\bV \bW \bV^T$ as
\begin{align*}
\bV \bW \bV^T &= \bV^* \bW (\bV^*)^T + \omega_n\left[ \sqrt{n}\, (\bV - \bV^*) \bW (\bV^*)^T + \sqrt{n}\, \bV^* \bW (\bV - \bV^*)^T \right] \\
&+ \omega_n^2 n\, (\bV - \bV^*) \bW  (\bV - \bV^*)^T =: \bT + \omega_n \bT^{(1)} + \omega_n^2 \bT^{(2)}.
\end{align*}
Note each column of $\bV - \bV^*$ is a sample average of $ \bff_\ell(\xx_i, y_i) - \E \bff_\ell(\xx_i, y_i)$, so $\sqrt{n}\, (\bV - \bV^*)$ is of the order $O_P(1)$. Also, since $\bW \xrightarrow{p} \bW^* \succ \bzero$, we deduce that with probability $1 - o(1)$, the rank of $\bV^* \bW (\bV^*)^T$ is $r$ due to Assumption~\ref{ass:span}, and that the projection matrix associated with its null space is exactly $\bP$. Following a similar argument as in \cite{Li91} and \cite{Li92}, by Lemma~\ref{lem:pert}, we deduce\footnote{We think this argument is not rigorous in the cited papers, as Lemma~\ref{lem:pert} is an asymptotic result for fixed matrices, whereas we substitute $\bT^{(1)}, \bT^{(2)}$ by random matrices. This issue, however, can be easily resolved; see the next proof.}
\begin{align}\label{eq:lambdabar}
\bar \lambda := \sum_{j=r+1}^n \lambda_j(\bV \bW \bV^T) = \omega_n \lambda^{(1)} + \omega_n^2 \lambda^{(2)} + o_P(\omega_n^2),
\end{align}
where $\lambda^{(1)} = \Tr\left( \bT^{(1)} \bP\right)$ and $\lambda^{(2)} = \Tr\left( \bT^{(2)} \bP - \bT^{(1)} \bT^+\bT^{(1)} \bP\right) $.  Note that the column vectors of $\bV^*$ (namely $\E \bv_j$s) lie in the subspace $\cS$, so $(\bV^*)^T \bP= \bzero$.  This yields
\begin{align*}
&\Tr\left( (\bV - \bV^*) \bW (\bV^*)^T \bP \right) = 0,  \quad \text{and} \\
&\Tr\left( \bV^* \bW (\bV - \bV^*)^T \bP \right) =\Tr\left( \bP \bV^* \bW (\bV - \bV^*)^T \right) = 0.
\end{align*}
Adding the above two equalities, we have $\Tr\left( \bT^{(1)} \bP\right) = 0$, which implies $\lambda^{(1)} = 0$. Thus, the dominant term in the expansion of $\bar \lambda$ is $\omega_n^2 \lambda^{(2)}$. Now we simplify $\bar \lambda$:
\begin{align*}
\bar \lambda &= \omega_n^2 \Tr\left( \bT^{(2)} \bP \right) - \omega_n^2 \Tr\left( \bT^{(1)} \bT^+\bT^{(1)} \bP \right) \\
&= \omega_n^2 \Tr\left( \bT^{(2)} \bP \right) - \omega_n^2 \Tr\left(\bP \bT^{(1)} \bT^+\bT^{(1)} \bP \right) \\
&= \Tr \left( \bP(\bV - \bV^*)  \bW  (\bV - \bV^*)^T\bP \right) - \Tr \left( \bP(\bV - \bV^*)  \tilde \bW  (\bV - \bV^*)^T\bP \right),
\end{align*}
where $\tilde \bW:= \bW (\bV^*)^T \bT^+ \bV^* \bW$. In the second line, we used the identities $\bP^2 = \bP$ and $\Tr(\bA \bB) = \Tr(\bB \bA)$; and in the third line, we used the equalities
\begin{equation*}
\bT^{(1)} \bP = \sqrt{n}\, \bV^* \bW (\bV - \bV^*)^T \bP, \quad \bP \bT^{(1)} =  \sqrt{n}\, \bP (\bV - \bV^*) \bW (\bV^*)^T.
\end{equation*}


Next we denote $\bZ = (\bV - \bV^*)^T\bP = \bV^T \bP  \in \R^{m \times p}$, and simplify the above expression using Lemma~\ref{lem:Kron} (i$\mspace{1mu}\mathrm{x}$):
\begin{equation*}
\bar \lambda = \Tr\left( \bZ^T (\bW - \tilde \bW) \bZ \right) = (\vect(\bZ))^T \left[ \bI_p \otimes (\bW - \tilde \bW) \right] \vect(\bZ)
\end{equation*}
Note that $\bZ^T = \bP \bV$ and $\vect(\bZ^T) = \bg(\bU^*)$, so by the central limit theorem, $\sqrt{n}\,\vect(\bZ^T) \xrightarrow{d} N(\bzero, \bar \bSigma)$ where $\bar \bSigma = \bar \bP^*\bar \bS^* \bar \bP^*$ (recall $\bar \bP^* = \bI_p \otimes \bP$). Similarly, $\sqrt{n}\, \vect(\bZ) \xrightarrow{d} N(\bzero, \bar \bSigma')$, where $\bar \bSigma'$ is a row-wise and column-wise permuted version of $\bar \bSigma$ (Lemma~\ref{lem:Kron}  (vii)). Also, by Assumption~\ref{ass:W}, we have $\bW \xrightarrow{p} \bW^*$; moreover, we claim
\begin{equation*}
\tilde \bW \xrightarrow{p} \tilde \bW^*, \quad \text{where} \quad \tilde \bW^* =  \bW^* (\bV^*)^T \left[ \bV^* \bW^* (\bV^*)^T \right]^+ \bV^* \bW^*.
\end{equation*}
In fact, although the Moore-Penrose pseudoinverse is not a continuous map in general, we still have $[ \bV^* \bW (\bV^*)^T ]^+ \xrightarrow{p} [ \bV^* \bW^* (\bV^*)^T ]^+$, because with probability $1-o(1)$, the matrix $(\bV^*)^T \bW \bV^*$ has exactly rank $r$.
Thus, by the (multivariate) Slutsky's theorem,
\begin{equation}\label{asym:lambdabar}
n\bar \lambda \xrightarrow{d} \bxi^T (\bar \bSigma')^{1/2} \left[ \bI_p \otimes (\bW^* - \tilde \bW^*) \right] (\bar \bSigma')^{1/2} \bxi,
\end{equation}
where $\bxi \sim N(\bzero, \bI_{pm})$. This implies $\bar \lambda = O_P(1/n)$ and thus proves the first claim of Theorem~\ref{thm:dimest}. Moreover, assuming Assumption~\ref{ass:uniform} in addition, we have
\begin{equation*}
\bar \bSigma = (\bI_p \otimes \bP) (\bS^* \otimes \bI_p) (\bI_p \otimes \bP) =  \bS^* \otimes \bP,
\end{equation*}
due to Lemma~\ref{lem:Kron} (i). We can proceed to simplify the right-hand side (RHS) of \eqref{asym:lambdabar}. $\bar \bSigma' = \bP \otimes \bS^*$ by Lemma~\ref{lem:Kron} (vii) and (viii), and thus $(\bar \bSigma')^{1/2} = \bP \otimes (\bS^*)^{1/2}$ by Lemma~\ref{lem:Kron} (i). Therefore,
\begin{align}
\text{RHS of \eqref{asym:lambdabar}} &= \bxi^T \left[ \bP \otimes  \left( (\bS^*)^{1/2} ( \bW^* - \tilde \bW^*) (\bS^*)^{1/2} \right) \right] \bxi =: \bxi^T \bM\, \bxi. \label{eq:RHS}
\end{align}
By Lemma~\ref{lem:sigmaS}, $\bSigma^* = (p-r) \bS^*$, so with the choice $\bW^* = (\bSigma^*)^{-1}$, we derive $(\bS^*)^{1/2}\bW^*(\bS^*)^{1/2} = \frac{1}{p-r} \bI_m $. We also rewrite $(\bS^*)^{1/2} \tilde \bW^*(\bS^*)^{1/2}$ as
\begin{equation*}
(\bS^*)^{1/2} \tilde \bW^* (\bS^*)^{1/2} = \frac{1}{p-r} \, \bK \left[ \bK^T \bK \right]^+ \bK^T,
\end{equation*}
where $\bK := (\bS^*)^{-1/2} (\bV^*)^T \in \R^{m \times p}$. The matrix $\bK$ has rank $r$, since $(\bS^*)^{-1/2}$ has full rank and $\bV^*$ has  rank $r$. Therefore, by Lemma~\ref{lem:pseudo}, $(p-r)(\bS^*)^{1/2} \tilde \bW^* (\bS^*)^{1/2}$ is a projection matrix with rank $r$, and consequently, $(p-r)\,(\bS^*)^{1/2} ( \bW^* - \tilde \bW^*) (\bS^*)^{1/2} = \bI_m - (\bS^*)^{1/2} \tilde \bW^* (\bS^*)^{1/2}$ is a projection matrix with rank $m - r$.

It follows from Lemma~\ref{lem:Kron} (iv) and (v) that $(p-r)\bM$, where $\bM$ is defined in \eqref{eq:RHS}, is a projection matrix with rank $(p-r)(m-r)$. Finally, by Thm.\ 2.7 of \cite{Seb03}, we conclude that $(p-r)\bxi^T \bM \, \bxi$ follows a chi-squared distribution $\chi^2_{(p-r)(m-r)}$, which finishes the proof.
\end{proof}

\begin{proof}[\bf Proof of \eqref{eq:lambdabar}]
To derive the above asymptotic result rigorously, we use the notations in \cite{Kato66} and invoke (3.2), (3.3), (3.4) and (3.6) in Chap.\ Two of \cite{Kato66} to bound the residual term in the asymptotic expansion. Under Assumption~\ref{ass:span}, there is a constant gap between $\lambda_r(\bV^* \bW^* (\bV^*))$ and $\lambda_{r+1}(\bV^* \bW^* (\bV^*))$. Since $\bW \xrightarrow{p} \bW^*$, we deduce that with probability $1-o(1)$, the gap between $\lambda_r(\bV^* \bW (\bV^*))$ and $\lambda_{r+1}(\bV^* \bW (\bV^*))$ is bounded away from zero by a constant. Thus, choosing any circle with constant radius enclosing $\lambda=0$, we have with probability $1-o(1)$,  $\rho$ in (3.4) and $\max_{\zeta \in \Gamma}\| R(\zeta) \|$ and both bounded by a constant. Since $\| \bT^{(1)} \| = O_P(1)$, $\| \bT^{(2)} \| = O_P(2)$ and $\bT^{(n)}$ ($n\ge3$) are simply zero, (3.2) is satisfied with probability $1-o(1)$ for $|x| < r_n$ where $r_n>0$ is any vanishing sequence. This implies with probability $1-o(1)$, $r_0$ in (3.3) satisfies $r_0\ge r_n$. We fix the `$n$' in (3.6) by $2$ (not to confuse with our sample size $n$). Then, the upper bound in (3.6) is $O(n^{-3/2} / r_0^3) = O(n^{-1})$ with the choice, say, $r_n = n^{-1/10}$. Therefore, we conclude that the residual term in the second-order expansion of $\lambda(\omega_n)$ is $o_P(\omega_n^2)$.
\end{proof}

\begin{proof}[\bf Proof of Corollary~\ref{cor:dimest}]
Let us drop the subscript $n$ as usual and denote $\E \bV$ by $\bV^*$. First, we claim that there exists a constant $c>0$ such that $\lambda_r(\bV^* (\bSigma^*)^{-1} (\bV^* )^T) > c $. In fact, each column of $\bV^*$ lies in $\cS$, so we can rewrite $\bV^*$ as $\bU^* (\bG^*)^T$ where we recall $\bG^* = (\bV^*)^T \bU^* \in \R^{m \times r}$. Under Assumption~\ref{ass:span}, $\bG^*$ has full (column) rank, so $\bG^* \xx \neq \bzero$ for any $\xx \neq \bzero$, which implies invertibility of $(\bG^*)^T  (\bSigma^*)^{-1} \bG^*$. Since
\begin{equation*}
\bV^* (\bSigma^*)^{-1} (\bV^* )^T = \bU^* \left(  (\bG^*)^T  (\bSigma^*)^{-1} \bG^* \right)  (\bU^*)^T,
\end{equation*}
we deduce that top $r$ eigenvalues of $\bV^* (\bSigma^*)^{-1} (\bV^* )^T$ are nonzero, and thus bounded below by some constant $c>0$. Using this claim, together with the fact that $\bV \bW \bV^T \xrightarrow{p} \bV^* (\bSigma^*)^{-1} (\bV^* )^T$ and the condition $\tau_n = o(1)$, we obtain $\P( \hat r_{\tau} \ge r) \to 1$ as $n \to \infty$. On the other hand, Theorem~\ref{thm:dimest} implies that for all $j>r$, the eigenvalue $\lambda_j(\bV \bW \bV^T)$ is $O_P(n^{-1})$, so the condition $n\tau_n \to \infty$ ensures that for all $j$, $\lambda_j(\bV \bW \bV^T) \le \tau_n$ with probability $1-o(1)$. This leads to $\P( \hat r_{\tau} \le r) \to 1$, and thus completing the proof of the first part of the corollary. The second part follows similarly.
\end{proof}

\subsection*{Proofs for Section~\ref{sec:singular}}
The omit the proof of Corollary~\ref{cor:opt3}, which is almost identical to that of Corollary~\ref{cor:procedure}.
\begin{proof}[\bf Proof of Theorem~\ref{thm:opt3}]
Let us break the proof into several steps.\\
\textbf{(1)} We prove a result parallel to \eqref{ineq:optimal}. Fix any $\bW^* \in \cW$. We will prove that $(\bG^*)^T (\bSigma^*)^{+} \bG^*$ is invertible, and that
\begin{equation}\label{ineq:optimal3}
[(\bG^*)^T \bW^* \bG^* ]^{-1} (\bG^*)^T \bW^* \bSigma^* \bW^* \bG^* [(\bG^*)^T \bW^* \bG^* ]^{-1} \succeq [(\bG^*)^T (\bSigma^*)^{+} \bG^* ]^{-1}.
\end{equation}
This essentially replaces $(\bSigma^*)^{-1}$ in \eqref{ineq:optimal} by $(\bSigma^*)^{+}$, and is equivalent to
\begin{equation*}
(\bG^*)^T \bW^* \bSigma^* \bW^* \bG^*  \succeq (\bG^*)^T \bW^* \bG^* [(\bG^*)^T (\bSigma^*)^{+} \bG^* ]^{-1} (\bG^*)^T \bW^* \bG^*.
\end{equation*}
It suffices, therefore, to prove invertibility of $(\bG^*)^T (\bSigma^*)^{+} \bG^*$ and
\begin{equation*}
 \bSigma^*  \succeq  \bG^* [(\bG^*)^T (\bSigma^*)^{+} \bG^* ]^{-1} (\bG^*)^T.
\end{equation*}
The right-hand side above can be simplified under Assumption~\ref{ass:redundancy}. To do so, we write
\begin{equation}\label{eq:bG}
\bG^* = \bB  \left( \begin{array}{c} \bG_1^* \\ \bzero \end{array} \right), \quad \text{where } \bB := \left( \begin{array}{cc} \bI_{m_1} & \bzero \\ \bSigma_{21}^* ( \bSigma_{11}^*)^{-1} & \bI_{m_2}\end{array} \right) \in \R^{m \times m}.
\end{equation}
We also express $\bSigma^*$ using $\bB$:
\begin{equation*}
\bSigma^* = \bB  \left( \begin{array}{cc} \bSigma_{11}^* & \bzero \\ \bzero & \bSigma_{22}^* - \bSigma_{21}^* (\bSigma_{11}^*)^{-1} \bSigma_{12}^* \end{array} \right)  \bB^T.
\end{equation*}
We make two observations: (a) the matrix $\bB$ is invertible; (b) $\bSigma_{22}^* - \bSigma_{21}^* (\bSigma_{11}^*)^{-1} \bSigma_{12}^* \succeq \bzero$ due to $\bSigma^* \succeq \bzero$. This yields
\begin{equation*}
(\bSigma^*)^+ = (\bB^T)^{-1}  \left( \begin{array}{cc} (\bSigma_{11}^*)^{-1} & \bzero \\ \bzero & (\bSigma_{22}^* - \bSigma_{21}^* (\bSigma_{11}^*)^{-1} \bSigma_{12}^*)^+ \end{array} \right)    \bB^{-1}.
\end{equation*}
This equality can be verified against the definition of Moore-Penrose pseudoinverse. Note that under Assumption~\ref{ass:redundancy}, $\bSigma_{11}^*$ is invertible while $\bSigma_{22}^* - \bSigma_{21}^* (\bSigma_{11}^*)^{-1} \bSigma_{12}^*$ may not. Now
\begin{align*}
(\bG^*)^T (\bSigma^*)^{+} \bG^* &= ( (\bG_1^*)^T~ \bzero) \left( \begin{array}{cc} (\bSigma_{11}^*)^{-1} & \bzero \\ \bzero & (\bSigma_{22}^* - \bSigma_{21}^* (\bSigma_{11}^*)^{-1} \bSigma_{12}^*)^+ \end{array} \right) \left( \begin{array}{c} \bG_1^* \\ \bzero \end{array} \right) \\
&= (\bG^*_1)^T (\bSigma_{11}^*)^{-1}\bG^*_1.
\end{align*}
By Assumption~\ref{ass:span}, $\bG^*$ has full column rank, and thus $\bG^*_1$ also has, due to \eqref{eq:bG}. It follows that $(\bG^*_1)^T (\bSigma_{11}^*)^{-1}\bG^*_1$ is positive definite (thus invertible). This proves the invertibility of $(\bG^*)^T (\bSigma^*)^{+} \bG^*$. Moreover, to prove \eqref{eq:bG}, it is equivalent to show
\begin{equation*}
\bB \left( \begin{array}{cc} \bSigma_{11}^* & \bzero \\ \bzero & \bSigma_{22}^* - \bSigma_{21}^* (\bSigma_{11}^*)^{-1} \bSigma_{12}^* \end{array} \right)  \bB^T \succeq \bB  \left( \begin{array}{c} \bG_1^* \\ \bzero \end{array} \right) \left[(\bG^*_1)^T (\bSigma_{11}^*)^{-1}\bG^*_1\right]^{-1} ( (\bG_1^*)^T, \bzero)  \bB^T,
\end{equation*}
which is also equivalent to show
\begin{equation*}
\bSigma_{11}^* \succeq \bG_1^* \left[(\bG^*_1)^T (\bSigma_{11}^*)^{-1}\bG^*_1\right]^{-1} (\bG_1^*)^T.
\end{equation*}
In the case of nonsingular weighting matrices, a similar inequality is known and is the key to the proof of \eqref{ineq:optimal}. Our above derivation reduces the singular case to the nonsingular case (note $\bSigma^*$ may be singular but $\bSigma_{11}^*$ is nonsingular), so the desired inequality follows.

\textbf{(2)} We prove \eqref{ineq:opt3}. First we observe that \eqref{eq:useful1} still holds---see the sentence in parentheses before \eqref{eq:useful1-0}. This leads to the same expression \eqref{eq:useful1} as in the proof of Theorem~\ref{thm:opt2}. Likewise, we still have \eqref{conv:AP} and  \eqref{eq:aePhi}, because $(\bG^*)^T \bW^* \bG^*$ is invertible, and $\hat \bU$ (associated with $\bW^*$) is consistent up to rotation by Theorem~\ref{thm:consist2}.

Thus, for any $\bW^* \in \cW $ with $\bW \xrightarrow{p} \bW^*$, the expression for $\aE(\bPsi(\bW))$ is given by \eqref{eq:aePhi}. In particular, if $\bW \xrightarrow{p} (\bSigma^*)^+$, then we have
\begin{align*}
\aE(\bPsi(\bW)) &= \frac{1}{n} \bA^* \bSigma^* (\bA^*)^T \\
&= \frac{1}{n} [(\bG^*)^T (\bSigma^*)^+ \bG^* ]^{-1} (\bG^*)^T (\bSigma^*)^+ \bSigma^* (\bSigma^*)^+ \bG^* [(\bG^*)^T (\bSigma^*)^+ \bG^* ]^{-1} \\
&=  \frac{1}{n} [(\bG^*)^T (\bSigma^*)^+ \bG^* ]^{-1},
\end{align*}
where we used the identity $(\bSigma^*)^+ \bSigma^* (\bSigma^*)^+ = (\bSigma^*)^+$ by definition of pseudoinverse. This proves the two equalities in \eqref{ineq:opt3}. The inequality in \eqref{ineq:opt3} is proved in part (1).

\textbf{(3)} Now we prove the final claim of the theorem. Let the eigen-decomposition of $\bSigma^*$ be $\bSigma^* = \bU_0 \bLambda_0 \bU_0^T$, where $\bU_0$ has orthonormal columns, and $\bLambda_0$ is a diagonal matrix consisting of all nonzero eigenvalues (the size of $\bLambda_0$ is possibly smaller than that of $\bSigma^*$). A basic property of the Moore-Penrose pseudoinverse is its representation via the eigen-decomposition: $(\bSigma^*)^+ = \bU_0 \bLambda_0^{-1} \bU_0^T$ \citep{Gol13}.

By the central limit theorem, we have $\| \hat \bSigma - \bSigma^* \|_2 = O_P(n^{-1/2})$. Recall the eigen-decomposition of $\hat \bSigma$ and \eqref{def:generalinv}. Weyl's inequality implies $| \bar \lambda_j - \lambda_j^*| = O_P(n^{-1/2})$, where $ \lambda_j^* $ is the $j$th largest eigenvalue of $\bSigma^*$. If $ \lambda_j^* > 0 $, then $\P(\bar \lambda_j > \delta_n) = 1 - o(1)$ due to $\delta_n = o(1)$, and thus
\begin{equation}\label{ineq:barlambda}
| \psi(\bar \lambda_j) - (\lambda_j^*)^{-1} | = O_P(n^{-1/2}).
\end{equation}
If $ \lambda_j^* = 0 $, then $\P(\bar \lambda_j \le \delta_n) = 1 - o(1)$ due to $\sqrt{n}\, \delta_n \to \infty$, and therefore $\psi(\bar \lambda_j) = 0$ with probability $1-o(1)$. This proves consistency of $\psi(\bar \lambda_1), \ldots, \psi(\bar \lambda_m)$ in \eqref{def:generalinv}.

Next, let $\bU_0'$ be a submatrix of $\bU_0$ such that the columns of $\bU_0'$ correspond to the same eigenvalue of $\bSigma^*$ (taking into account eigenvalue multiplicity). And denote by $\bar \bU'$ the counterpart of $\bar \bU$. by Davis-Kahan's theorem \citep{DavKah70} and a $\sin\Theta$ formula \citep{Ste90}
\begin{equation}\label{ineq:barU}
\| \bar \bU' (\bar \bU')^T - \bU_0' (\bU_0')^T \|_2 = O(\| \hat \bSigma - \bSigma^* \|_2) = O_P(n^{-1/2}).
\end{equation}
We can write $(\bSigma^*)^+$ as a sum of the form $\lambda^{-1} \bU_0' (\bU_0')^T$, where $\lambda$ is any positive eigenvalue of $\bSigma^*$ and $\bU_0'$ corresponds to $\lambda$. We can do the same for $\bW$, and combine the bounds in \eqref{ineq:barlambda} and \eqref{ineq:barU} to deduce $\bW \xrightarrow{p} (\bSigma^*)^+$. This completes the proof.
\end{proof}

\subsection*{Proofs for Section~\ref{sec:augment}}

\begin{proof}[{\bf Proof of Theorem~\ref{thm:augment}}]
We shall denote by $\bM_n^r$ is the best rank-$r$ approximation of $n^{-1} \bX^T \bX - \sigma^2 \bI_p$ (i.e., in the eigen-decomposition, only keeping top $r$ terms with largest absolute eigenvalues). We also define
\begin{equation*}
\tilde \bU_n^0 := \tilde \bU_n^0(\kappa,\bW_n') = \eigen_r\left(\kappa \bM^r_n +  \bV_n \bW_n' \bV_n^T\right).
\end{equation*}
We will prove 
\begin{align}
\aE\left( \left\| \hat \bU_n^{\GMM} [\hat \bU_n^{\GMM}]^T - \bU^* (\bU^*)^T  \right\|_F^2  \right)& \le \aE\left( \left\| \tilde \bU^0_n [\tilde \bU^0_n]^T - \bU^* (\bU^*)^T  \right\|_F^2  \right) \notag \\
& = \aE\left( \left\| \tilde \bU_n \tilde \bU_n^T - \bU^* (\bU^*)^T  \right\|_F^2  \right). \label{ineq:augment}
\end{align}

Without loss of generality, we assume that $\bM$ is formed by the first $p$ columns of $\bV$. By an elementary property of pseudoinverse, we have $\bM^r = \bM (\bM^r)^+ \bM^T$. This identity can be proved by, for example, using the eigen-decomposition of $\bM$ and rewriting $\bM^r$ and $(\bM^r)^+$ \citep[Thm. 2.4.8, Chap. 5.5.2]{Gol13}.

{\bf (1)} First we prove the inequality in \eqref{ineq:augment}. Using the above identity of $\bM^r$, we can absorb the term $\kappa \bM^r$ into $\bV \bW' \bV^T$ as follows.
\begin{equation*}
\kappa \bM^r +  \bV \bW' \bV^T = \bV \left[ \left( 
\begin{array}{cc} \kappa(\bM^r)^+ & \bzero \\ \bzero & \bzero 
\end{array} \right)  + \bW' \right] \bV^T =: \bV \bW^{\mathrm{new}} \bV^T 
\end{equation*}
We claim that $(\bM^r)^+ \xrightarrow{p} (\E \bM)^+ =  \bU^* \bLambda^* (\bU^*)^T$, where $\bLambda^* \in \R^{r \times r}$ is a diagonal matrix with positive diagonal entries. This can be proved in a way similar to part (3) of the proof of Theorem~\ref{thm:opt3}. In fact, since $\| \bM - \E \bM \|_2 = O_P(n^{-1/2})$, by Weyl's inequality, the eigenvalues have convergence $\lambda_j(\bM) \xrightarrow{p} \lambda_j(\E \bM) $ for all $j \in [p]$, where $\lambda_j(\cdot)$ denotes the $j$th largest eigenvalue of a matrix. For $j \le r$, the assumption implies $\lambda_j(\E \bM) > 0$, so we also have convergence $[\lambda_j(\bM)]^{-1} \xrightarrow{p} [\lambda_j(\E \bM)]^{-1}$. Moreover, by Davis-Kahan's theorem \citep{DavKah70}, the corresponding eigenspace of $\lambda_j(\bM)$ also converges in probability to that of $\lambda_j(\E \bM)$, where $j \le r$. Therefore, $(\bM^r)^+ \xrightarrow{p} (\E \bM)^+$, and due to $\rank(\E \bM) = r$ and $\spann(\E \bM) = \spann(\bU^*)$, we can express $(\E \bM)^+$ as $\bU^* \bLambda^* (\bU^*)^T$.

Using this claim, we deduce that $\bW^{\mathrm{new}}$ converges in probability to a semidefinite matrix. Since $\bW' \xrightarrow{p} \bW^* \in \cW$, the limit of $\bW^{\mathrm{new}}$ must be also in $\cW$. Thus, we can treat $\bW^{\mathrm{new}}$ as the new weighting matrix $\bW'$ in Corollary~\ref{cor:opt3}, and the inequality in \eqref{ineq:augment} follows from the conclusion of Corollary~\ref{cor:opt3}.

{\bf (2)} Denote $\bY = \kappa\bM^r + \bV \bW' \bV^T$ and $\bDelta = \bM - \bM^r$. Note that $\bDelta = \sum_{j>r} \lambda_j(\bM) \bu_j(\bM) [\bu_j(\bM)]^T$, where $\bu_j(\bM)$ is the eigenvector of $\bM$ corresponding to $\lambda_j(\bM)$. By the assumption $\rank(\E \bM) = r$ and the analysis in part (1), we have $|\lambda_j(\bM)| = O_P(n^{-1/2})$ and $(\bU^*)^T \bu_j(\bM) \xrightarrow{p} \bzero$ for $j > r$. Applying the consistency result (Theorem~\ref{thm:consist1}) to $\tilde \bU^0$, we also obtain $(\tilde \bU^0)^T \bu_j(\bM) \xrightarrow{p} \bzero$.

Now, we view $\bDelta$ as a perturbation matrix added to $\bY$, and view $\tilde \bU$ as a resulting perturbed version of $\tilde \bU^0$. We use (the original version of) Davis-Kahan's theorem \citep{DavKah70} to obtain
\begin{equation*}
\| \tilde \bU (\tilde \bU)^T - \tilde \bU^0 (\tilde \bU^0)^T \|_F = O_P\left( \|  \bDelta  \tilde \bU^0 \|_F  \right). 
\end{equation*}
This form of Davis-Kahan's theorem has appeared in recent works \citep{Yu14, ZhoBou18}. Using $|\lambda_j(\bM)| = O_P(n^{-1/2})$ and $(\tilde \bU^0)^T \bu_j(\bM) \xrightarrow{p} \bzero$ we established, we obtain
\begin{equation}\label{ineq:smallerorder}
\| \tilde \bU (\tilde \bU)^T - \tilde \bU^0 (\tilde \bU^0)^T \|_F = o_P(n^{-1/2}).
\end{equation}
Note that Theorem~\ref{thm:opt3} implies that $\sqrt{n}\, \| \tilde \bU_n^0  [\tilde \bU_n^0]^T - \bU^* (\bU^*)^T  \|_F$ converges in distribution to a non-vanishing limit. Thus, from \eqref{ineq:smallerorder}, we deduce that $\sqrt{n}\, \| \tilde \bU_n \tilde \bU_n^T - \bU^* (\bU^*)^T  \|_F$ must also converge in distribution to the same limit, and therefore, the equality in \eqref{ineq:augment} is true.
\end{proof}

\footnotesize{
\bibliographystyle{ims}
\bibliography{ref2}
}

\end{document}